\documentclass[a4paper,12pt]{article}

\usepackage[T1]{fontenc}
\usepackage{subcaption}
\usepackage{amsmath}
\usepackage{amsbsy}
\usepackage{amsthm}
\usepackage{amsfonts}
\usepackage{amssymb}
\usepackage{mathtools}
\usepackage{array}
\usepackage{bbm} %
\usepackage{tikz}
\usepackage{newunicodechar}
\newunicodechar{ȩ}{\c{e}} %

\usepackage[style=alphabetic,url=false,doi=false,backend=biber,maxbibnames=20]{biblatex}

\AtEveryBibitem{\clearfield{number}}
\AtEveryBibitem{\clearfield{issn}}

\reversemarginpar

\usetikzlibrary{math}
\usetikzlibrary{calc}

\usepackage{hyperref}
\usetikzlibrary{external}
\makeatletter
\@ifpackageloaded{hyperref}{ %
	}{
	\newcommand{\href}[2]{#2}
	\tikzexternalize
	\renewcommand{\todo}[2][]{\tikzexternaldisable\@todo[#1]{#2}\tikzexternalenable}
}
\makeatother

\usepackage[noabbrev,capitalize]{cleveref}
\usepackage{lastpage}

\makeatletter 
\tikzset{curlybrace/.style={rounded corners=2pt,line cap=round}}%
\pgfkeys{%
/curlybrace/.cd,%
tip angle/.code     =  \def\cb@angle{#1},
/curlybrace/.unknown/.code ={\let\searchname=\pgfkeyscurrentname
                              \pgfkeysalso{\searchname/.try=#1,
                              /tikz/\searchname/.retry=#1}}}  
\def\curlybrace{\pgfutil@ifnextchar[{\curly@brace}{\curly@brace[]}}%

\def\curly@brace[#1]#2#3#4{%
\pgfkeys{/curlybrace/.cd,
tip angle = 0.75}%
\pgfqkeys{/curlybrace}{#1}%
\ifnum 1>#4 \def\cbrd{0.05} \else \def\cbrd{0.075} \fi
\draw[/curlybrace/.cd,curlybrace,#1]  (#2:#4-\cbrd) -- (#2:#4) arc (#2:{(#2+#3)/2-\cb@angle}:#4) --({(#2+#3)/2}:#4+\cbrd) coordinate (curlybracetipn);
\draw[/curlybrace/.cd,curlybrace,#1] ({(#2+#3)/2}:#4+\cbrd) -- ({(#2+#3)/2+\cb@angle}:#4) arc ({(#2+#3)/2+\cb@angle} :#3:#4) --(#3:#4-\cbrd);
}

\tikzset{vertex/.style={circle,fill=black,inner sep=1pt}
	,cutline/.style={blue,very thick}
	,cutbond/.style={red}
}

\newcounter{Ccounter}
\newcommand{\Clast}{\ensuremath{C_{\theCcounter}}}
\newcommand{\C}{
	{\refstepcounter{Ccounter}
	\Clast}
}
\newcommand{\Cl}[1]{
	{\refstepcounter{Ccounter}
	\label{C:#1}
	\Clast}
}

\@ifpackageloaded{hyperref}{
	\newcommand{\Cr}[1]{\ensuremath{C_{\ref*{C:#1}}}}
}{ %
	\newcommand{\Cr}[1]{\ensuremath{C_{\ref{C:#1}}}}
}

\let\originalleft\left
\let\originalright\right
\renewcommand{\left}{\mathopen{}\mathclose\bgroup\originalleft}
\renewcommand{\right}{\aftergroup\egroup\originalright}

\newtheorem{theorem}{Theorem}
\newtheorem{lemma}[theorem]{Lemma}

\newtheorem{cor}[theorem]{Corollary}
\newtheorem{corollary}[theorem]{Corollary}

\newtheorem{definition}[theorem]{Definition}

\newcommand{\wt}{\widetilde}

\newcommand{\comp}{\ensuremath{\mathsf{c}} }
\renewcommand{\complement}{^\mathsf{c}}
\newcommand{\ind}{\mathbbm{1}}
\DeclareMathOperator{\diam}{diam}
\DeclareMathOperator{\dist}{dist}
\DeclareMathOperator{\erfc}{erfc}
\DeclareMathOperator{\Tr}{Tr}

\DeclareMathOperator{\spec}{spec}
\DeclareMathOperator{\ceil}{ceil}

\DeclareMathOperator{\supp}{supp}

\newcommand{\dd}{\,\text{\rm d}}             %

\renewcommand{\Re}{\operatorname{Re}}

\newcommand{\lis}[1]{{\ensuremath{\overline{#1}}}}

\newcommand{\B}{\ensuremath{\textup{B}} }
\newcommand{\D}{\ensuremath{\textup{D}} }
\newcommand{\E}{\ensuremath{\textup{E}} }
\newcommand{\R}{\ensuremath{\textup{R}} }
\newcommand{\T}{\ensuremath{\textup{T}} }

\newcommand{\bC}{\ensuremath{\mathbb{C}}}
\newcommand{\bE}{\ensuremath{\mathbb{E}}}

\newcommand{\bN}{\ensuremath{\mathbb{N}}}
\newcommand{\bP}{\ensuremath{\mathbb{P}}}

\newcommand{\bR}{\ensuremath{\mathbb{R}}}

\newcommand{\bZ}{\ensuremath{\mathbb{Z}}}

\newcommand{\cB}{\ensuremath{\mathcal{B}}}
\newcommand{\cC}{\ensuremath{\mathcal{C}}}

\newcommand{\cE}{\ensuremath{\mathcal{E}}}

\newcommand{\cG}{\ensuremath{\mathcal{G}}}

\newcommand{\cO}{\ensuremath{\mathcal{O}}}

\newcommand{\cU}{\ensuremath{\mathcal{U}}}

\newcommand{\cW}{\ensuremath{\mathcal{W}}}

\newcommand{\fC}{\ensuremath{\mathbf{C}}}
\newcommand{\fE}{\ensuremath{\mathbf{E}}}

\newcommand{\detS}{\det\nolimits_\Sigma}
\newcommand{\detLS}{\det\nolimits_{L\Sigma}}
\newcommand{\detsub}[1]{\det\nolimits_{#1}}

\newcommand{\email}[1]{\href{mailto:#1}{\texttt{#1}}}

\title{Discrete and zeta-regularized determinants of the Laplacian on polygonal domains with Dirichlet boundary conditions}
\author{Rafael L.\ Greenblatt\thanks{e-mail address: \email{greenblatt@mat.uniroma2.it}}
\\ SISSA, Trieste, Italy\thanks{Current affiliation: Dipartimento di Matematica, Universit\`a degli Studi di Roma ``Tor Vergata'', Rome, Italy } } 
\date{\today}

\begin{document}

\maketitle

\begin{abstract}
	For $\Pi \subset \mathbb{R}^2$ a connected, open, bounded set whose boundary is a finite union of disjoint polygons whose vertices have integer coordinates, 
	the logarithm of the discrete Laplacian on $L\Pi \cap \mathbb{Z}^2$ with Dirichlet boundary conditions has an asymptotic expression for large $L$ involving the zeta-regularized determinant of the associated continuum Laplacian.

	When $\Pi$ is not simply connected, this result extends to Laplacians acting on two-valued functions with a specified monodromy class.
\end{abstract}

\section{Introduction}
\label{sec:intro}

For a domain $\Omega \subset \bR^2$, let $\cG(\Omega)$ be the graph whose vertex set is ${\wt \Omega := \Omega \cap \bZ^2}$ and whose edge set $\cE(\Omega)$ is the set of pairs $\left\{ x,y \right\} \subset \wt \Omega$ such that the line segment $\overline{xy}$ has length one and is entirely contained in $\Omega$ (see Figure~\ref{fig:Pi_example} for examples).
The discrete Laplacian on $\Omega$ with Dirichlet boundary conditions is the operator $\wt\Delta_\Omega$ on $\ell^2(\Omega \cap \bZ^2)$ given by
\begin{equation}
	\wt\Delta_\Omega f(x)
	=
	4 f(x)
	- \sum_{\substack{y \in \bZ^2 \\ \overline{xy} \subset \Omega \\ |x-y|_1=1}} 
	f(y)
	,
	\label{eq:Lap}
\end{equation}
where $\overline{xy}$ is the closed line segment from $x$ to $y$.
Note that for bounded $\Omega$, this is a symmetric matrix with eigenvalues in $(0,4)$.
The main focus of this manuscript is studying $\det \wt\Delta_\Omega$ of the following form:  
we will fix a bounded open set $\Pi \subset \bR^2$ whose boundary is 
a disjoint union of polygons 
whose vertices are all in $\bZ^2$, and consider $\Omega = L \Pi$ for integer $L$.
\begin{figure}[h]
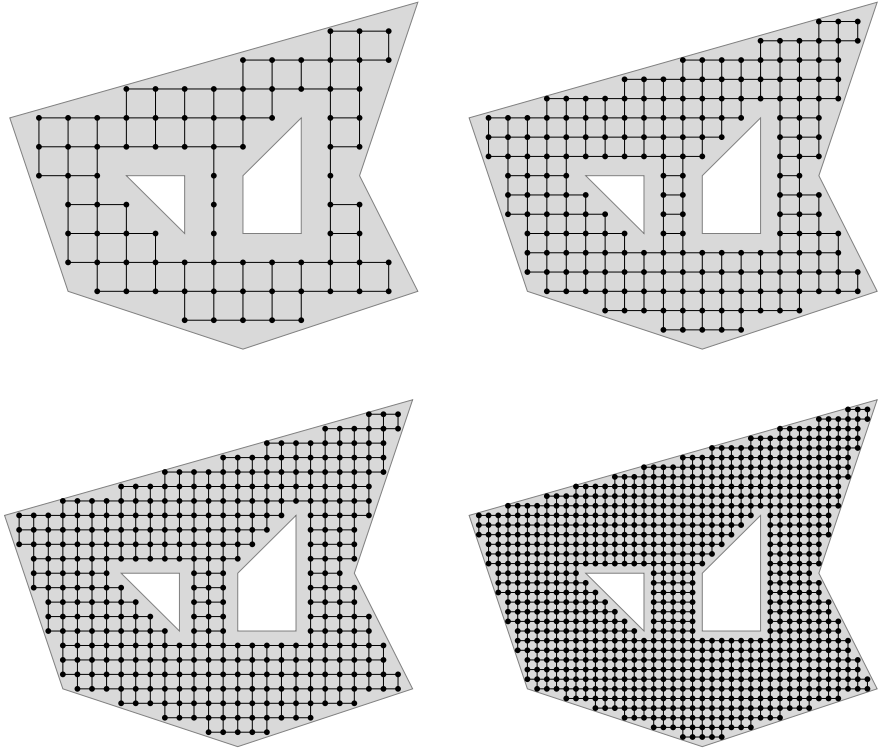

	\centering
	\begin{tabular}{cc}
		\newsavebox{\figbox}
		\begin{lrbox}{\figbox}
			\tikzpicturedependsonfile{PI2.tikz}
			\begin{tikzpicture}\draw[gray,fill=gray!30,even odd rule] (0,-1) -- (3,0) -- (2,2) -- (3,5) -- (-4,3) -- (-3,0) -- cycle
	(1,1) -- (1,3) -- (0,2) -- (0,1) -- cycle
	(-2,2) -- (-1,2) -- (-1,1) -- cycle;
\draw (-3.5,2.0) node[vertex] {};
\draw (-3.5,2.5) node[vertex] {};
\draw (-3.5,3.0) node[vertex] {};
\draw (-3.0,0.5) node[vertex] {};
\draw (-3.0,1.0) node[vertex] {};
\draw (-3.0,1.5) node[vertex] {};
\draw (-3.0,2.0) node[vertex] {};
\draw (-3.0,2.5) node[vertex] {};
\draw (-3.0,3.0) node[vertex] {};
\draw (-2.5,0.0) node[vertex] {};
\draw (-2.5,0.5) node[vertex] {};
\draw (-2.5,1.0) node[vertex] {};
\draw (-2.5,1.5) node[vertex] {};
\draw (-2.5,2.0) node[vertex] {};
\draw (-2.5,2.5) node[vertex] {};
\draw (-2.5,3.0) node[vertex] {};
\draw (-2.0,0.0) node[vertex] {};
\draw (-2.0,0.5) node[vertex] {};
\draw (-2.0,1.0) node[vertex] {};
\draw (-2.0,1.5) node[vertex] {};
\draw (-2.0,2.5) node[vertex] {};
\draw (-2.0,3.0) node[vertex] {};
\draw (-2.0,3.5) node[vertex] {};
\draw (-1.5,0.0) node[vertex] {};
\draw (-1.5,0.5) node[vertex] {};
\draw (-1.5,1.0) node[vertex] {};
\draw (-1.5,2.5) node[vertex] {};
\draw (-1.5,3.0) node[vertex] {};
\draw (-1.5,3.5) node[vertex] {};
\draw (-1.0,-0.5) node[vertex] {};
\draw (-1.0,0.0) node[vertex] {};
\draw (-1.0,0.5) node[vertex] {};
\draw (-1.0,2.5) node[vertex] {};
\draw (-1.0,3.0) node[vertex] {};
\draw (-1.0,3.5) node[vertex] {};
\draw (-0.5,-0.5) node[vertex] {};
\draw (-0.5,0.0) node[vertex] {};
\draw (-0.5,0.5) node[vertex] {};
\draw (-0.5,1.0) node[vertex] {};
\draw (-0.5,1.5) node[vertex] {};
\draw (-0.5,2.0) node[vertex] {};
\draw (-0.5,2.5) node[vertex] {};
\draw (-0.5,3.0) node[vertex] {};
\draw (-0.5,3.5) node[vertex] {};
\draw (0.0,-0.5) node[vertex] {};
\draw (0.0,0.0) node[vertex] {};
\draw (0.0,0.5) node[vertex] {};
\draw (0.0,2.5) node[vertex] {};
\draw (0.0,3.0) node[vertex] {};
\draw (0.0,3.5) node[vertex] {};
\draw (0.0,4.0) node[vertex] {};
\draw (0.5,-0.5) node[vertex] {};
\draw (0.5,0.0) node[vertex] {};
\draw (0.5,0.5) node[vertex] {};
\draw (0.5,3.0) node[vertex] {};
\draw (0.5,3.5) node[vertex] {};
\draw (0.5,4.0) node[vertex] {};
\draw (1.0,-0.5) node[vertex] {};
\draw (1.0,0.0) node[vertex] {};
\draw (1.0,0.5) node[vertex] {};
\draw (1.0,3.5) node[vertex] {};
\draw (1.0,4.0) node[vertex] {};
\draw (1.5,0.0) node[vertex] {};
\draw (1.5,0.5) node[vertex] {};
\draw (1.5,1.0) node[vertex] {};
\draw (1.5,1.5) node[vertex] {};
\draw (1.5,2.0) node[vertex] {};
\draw (1.5,2.5) node[vertex] {};
\draw (1.5,3.0) node[vertex] {};
\draw (1.5,3.5) node[vertex] {};
\draw (1.5,4.0) node[vertex] {};
\draw (1.5,4.5) node[vertex] {};
\draw (2.0,0.0) node[vertex] {};
\draw (2.0,0.5) node[vertex] {};
\draw (2.0,1.0) node[vertex] {};
\draw (2.0,1.5) node[vertex] {};
\draw (2.0,2.5) node[vertex] {};
\draw (2.0,3.0) node[vertex] {};
\draw (2.0,3.5) node[vertex] {};
\draw (2.0,4.0) node[vertex] {};
\draw (2.0,4.5) node[vertex] {};
\draw (2.5,0.0) node[vertex] {};
\draw (2.5,0.5) node[vertex] {};
\draw (2.5,4.0) node[vertex] {};
\draw (2.5,4.5) node[vertex] {};
\draw (-3.5,2.0) -- (-3.0,2.0);
\draw (-3.5,2.0) -- (-3.5,2.5);
\draw (-3.5,2.5) -- (-3.0,2.5);
\draw (-3.5,2.5) -- (-3.5,3.0);
\draw (-3.5,3.0) -- (-3.0,3.0);
\draw (-3.0,0.5) -- (-2.5,0.5);
\draw (-3.0,0.5) -- (-3.0,1.0);
\draw (-3.0,1.0) -- (-2.5,1.0);
\draw (-3.0,1.0) -- (-3.0,1.5);
\draw (-3.0,1.5) -- (-2.5,1.5);
\draw (-3.0,1.5) -- (-3.0,2.0);
\draw (-3.0,2.0) -- (-2.5,2.0);
\draw (-3.0,2.0) -- (-3.0,2.5);
\draw (-3.0,2.5) -- (-2.5,2.5);
\draw (-3.0,2.5) -- (-3.0,3.0);
\draw (-3.0,3.0) -- (-2.5,3.0);
\draw (-2.5,0.0) -- (-2.0,0.0);
\draw (-2.5,0.0) -- (-2.5,0.5);
\draw (-2.5,0.5) -- (-2.0,0.5);
\draw (-2.5,0.5) -- (-2.5,1.0);
\draw (-2.5,1.0) -- (-2.0,1.0);
\draw (-2.5,1.0) -- (-2.5,1.5);
\draw (-2.5,1.5) -- (-2.0,1.5);
\draw (-2.5,1.5) -- (-2.5,2.0);
\draw (-2.5,2.0) -- (-2.5,2.5);
\draw (-2.5,2.5) -- (-2.0,2.5);
\draw (-2.5,2.5) -- (-2.5,3.0);
\draw (-2.5,3.0) -- (-2.0,3.0);
\draw (-2.0,0.0) -- (-1.5,0.0);
\draw (-2.0,0.0) -- (-2.0,0.5);
\draw (-2.0,0.5) -- (-1.5,0.5);
\draw (-2.0,0.5) -- (-2.0,1.0);
\draw (-2.0,1.0) -- (-1.5,1.0);
\draw (-2.0,1.0) -- (-2.0,1.5);
\draw (-2.0,2.5) -- (-1.5,2.5);
\draw (-2.0,2.5) -- (-2.0,3.0);
\draw (-2.0,3.0) -- (-1.5,3.0);
\draw (-2.0,3.0) -- (-2.0,3.5);
\draw (-2.0,3.5) -- (-1.5,3.5);
\draw (-1.5,0.0) -- (-1.0,0.0);
\draw (-1.5,0.0) -- (-1.5,0.5);
\draw (-1.5,0.5) -- (-1.0,0.5);
\draw (-1.5,0.5) -- (-1.5,1.0);
\draw (-1.5,2.5) -- (-1.0,2.5);
\draw (-1.5,2.5) -- (-1.5,3.0);
\draw (-1.5,3.0) -- (-1.0,3.0);
\draw (-1.5,3.0) -- (-1.5,3.5);
\draw (-1.5,3.5) -- (-1.0,3.5);
\draw (-1.0,-0.5) -- (-0.5,-0.5);
\draw (-1.0,-0.5) -- (-1.0,0.0);
\draw (-1.0,0.0) -- (-0.5,0.0);
\draw (-1.0,0.0) -- (-1.0,0.5);
\draw (-1.0,0.5) -- (-0.5,0.5);
\draw (-1.0,2.5) -- (-0.5,2.5);
\draw (-1.0,2.5) -- (-1.0,3.0);
\draw (-1.0,3.0) -- (-0.5,3.0);
\draw (-1.0,3.0) -- (-1.0,3.5);
\draw (-1.0,3.5) -- (-0.5,3.5);
\draw (-0.5,-0.5) -- (0.0,-0.5);
\draw (-0.5,-0.5) -- (-0.5,0.0);
\draw (-0.5,0.0) -- (0.0,0.0);
\draw (-0.5,0.0) -- (-0.5,0.5);
\draw (-0.5,0.5) -- (0.0,0.5);
\draw (-0.5,0.5) -- (-0.5,1.0);
\draw (-0.5,1.0) -- (-0.5,1.5);
\draw (-0.5,1.5) -- (-0.5,2.0);
\draw (-0.5,2.0) -- (-0.5,2.5);
\draw (-0.5,2.5) -- (0.0,2.5);
\draw (-0.5,2.5) -- (-0.5,3.0);
\draw (-0.5,3.0) -- (0.0,3.0);
\draw (-0.5,3.0) -- (-0.5,3.5);
\draw (-0.5,3.5) -- (0.0,3.5);
\draw (0.0,-0.5) -- (0.5,-0.5);
\draw (0.0,-0.5) -- (0.0,0.0);
\draw (0.0,0.0) -- (0.5,0.0);
\draw (0.0,0.0) -- (0.0,0.5);
\draw (0.0,0.5) -- (0.5,0.5);
\draw (0.0,2.5) -- (0.0,3.0);
\draw (0.0,3.0) -- (0.5,3.0);
\draw (0.0,3.0) -- (0.0,3.5);
\draw (0.0,3.5) -- (0.5,3.5);
\draw (0.0,3.5) -- (0.0,4.0);
\draw (0.0,4.0) -- (0.5,4.0);
\draw (0.5,-0.5) -- (1.0,-0.5);
\draw (0.5,-0.5) -- (0.5,0.0);
\draw (0.5,0.0) -- (1.0,0.0);
\draw (0.5,0.0) -- (0.5,0.5);
\draw (0.5,0.5) -- (1.0,0.5);
\draw (0.5,3.0) -- (0.5,3.5);
\draw (0.5,3.5) -- (1.0,3.5);
\draw (0.5,3.5) -- (0.5,4.0);
\draw (0.5,4.0) -- (1.0,4.0);
\draw (1.0,-0.5) -- (1.0,0.0);
\draw (1.0,0.0) -- (1.5,0.0);
\draw (1.0,0.0) -- (1.0,0.5);
\draw (1.0,0.5) -- (1.5,0.5);
\draw (1.0,3.5) -- (1.5,3.5);
\draw (1.0,3.5) -- (1.0,4.0);
\draw (1.0,4.0) -- (1.5,4.0);
\draw (1.5,0.0) -- (2.0,0.0);
\draw (1.5,0.0) -- (1.5,0.5);
\draw (1.5,0.5) -- (2.0,0.5);
\draw (1.5,0.5) -- (1.5,1.0);
\draw (1.5,1.0) -- (2.0,1.0);
\draw (1.5,1.0) -- (1.5,1.5);
\draw (1.5,1.5) -- (2.0,1.5);
\draw (1.5,1.5) -- (1.5,2.0);
\draw (1.5,2.0) -- (1.5,2.5);
\draw (1.5,2.5) -- (2.0,2.5);
\draw (1.5,2.5) -- (1.5,3.0);
\draw (1.5,3.0) -- (2.0,3.0);
\draw (1.5,3.0) -- (1.5,3.5);
\draw (1.5,3.5) -- (2.0,3.5);
\draw (1.5,3.5) -- (1.5,4.0);
\draw (1.5,4.0) -- (2.0,4.0);
\draw (1.5,4.0) -- (1.5,4.5);
\draw (1.5,4.5) -- (2.0,4.5);
\draw (2.0,0.0) -- (2.5,0.0);
\draw (2.0,0.0) -- (2.0,0.5);
\draw (2.0,0.5) -- (2.5,0.5);
\draw (2.0,0.5) -- (2.0,1.0);
\draw (2.0,1.0) -- (2.0,1.5);
\draw (2.0,2.5) -- (2.0,3.0);
\draw (2.0,3.0) -- (2.0,3.5);
\draw (2.0,3.5) -- (2.0,4.0);
\draw (2.0,4.0) -- (2.5,4.0);
\draw (2.0,4.0) -- (2.0,4.5);
\draw (2.0,4.5) -- (2.5,4.5);
\draw (2.5,0.0) -- (2.5,0.5);
\draw (2.5,4.0) -- (2.5,4.5);\end{tikzpicture}
		\end{lrbox}
		\resizebox{0.4\textwidth}{!}{\usebox\figbox}
		&
		\begin{lrbox}{\figbox}
			\tikzpicturedependsonfile{PI3.tikz}
			\input{PI3.tikz}
		\end{lrbox}
		\resizebox{0.4\textwidth}{!}{\usebox\figbox}
		\\
		\\
		\begin{lrbox}{\figbox}
			\tikzpicturedependsonfile{PI4.tikz}
			\input{PI4.tikz}
		\end{lrbox}
		\resizebox{0.4\textwidth}{!}{\usebox\figbox}
		&
		\begin{lrbox}{\figbox}
			\tikzpicturedependsonfile{PI6.tikz}
			\input{PI6.tikz}
		\end{lrbox}
		\resizebox{0.4\textwidth}{!}{\usebox\figbox}
	\end{tabular}
	\caption{Example of a shape $\Pi$ and the graphs $\cG(L\Pi)$ for different integer values of $L$ (not to scale).}
	\label{fig:Pi_example}
\end{figure}
Throughout this work, I will use $\Omega$ for a general domain and $\Pi$ for a bounded domain of this polygonal form.

In fact I will consider a somewhat more general Laplacian, sometimes called the twisted Laplacian.  
For an assignment of $\rho_{xy} \in \bC$ to each $x, y \in \Omega \cap \bZ^2$ with $|x-y|_1 = 1$ such that $\rho_{xy} = 1/\rho_{yx}$, we can define the discrete scalar Laplacian on $\Omega$ with Dirichlet boundary conditions and connection $\rho$ as
\begin{equation}
	\wt\Delta_{\Omega,\rho} f(x)
	=
	4 f(x)
	- \sum_{\substack{y \in \bZ^2 \\ \overline{xy} \subset \Omega \\  |x-y|_1=1}} 
	\rho_{xy}
	f(y).
	\label{eq:rho_lap}
\end{equation}
Writing out $\det \wt \Delta_{\Omega,\rho}$ by the Leibniz formula, $\rho$ appears in the form of products 
\begin{equation}
	\prod_{j=1}^n
	\rho_{y_j,y_{j+1}}
	\label{eq:homotopy}
\end{equation}
where $y_1,\dots,y_{n+1}$ is a sequence which has no repeated elements except $y_{n+1} \equiv y_1$.  These sequences can be naturally identified with simple closed curves; I will call the associated products monodromy factors.
I will consider connections with the following property: for some finite $\Sigma \subset \Omega\complement$, the monodromy factor is $-1$ if the associated closed curve winds around an odd number of elements of $\Sigma$ and $1$ otherwise.
From the above considerations, $\det \wt \Delta_{\Omega,\rho}$ is the same for all $\rho$ with this property for the same $\Sigma$, so we can let
\begin{equation}
	\detS \wt \Delta_\Omega
	:=
	\det \wt \Delta_{\Omega,\rho}
	\label{eq:det_Sig}
\end{equation}
for an arbitrary chosen such connection; it is easy to construct such a $\rho$ for any $\Sigma$.
In fact in this way we obtain all of the connections which are ``flat'' (that is, for which the monodromy factor is $1$ for all contractible curves) and where the monodromy factors take values only in $\{\pm 1\}$.

These determinants are the partition functions of certain sets of essential cycle-rooted spanning forests \cite{Ken11}.  
In the simply connected case (where $\Sigma$ plays no role) the set of forests in question is simply the set of spanning trees on the graph formed by adding a ``giant'' vertex to $\cG(\Omega)$ which is connected to all the vertices on the boundary, and this identity is a restatement of the Kirchoff matrix-tree theorem.
In this case, and also for the case when $\Omega$ is doubly connected and $\Sigma$ consists of a single point in the finite component of $\Omega\complement$, this partition function is in turn equal to the number of perfect matchings (configurations of the dimer model) on a ``Temperleyan'' graph constructed from $\cG(\Omega)$; however for higher-genus planar domains this correspondence involves a set of forests which has a different characterization \cite{BLR}.

\begin{figure}[h]
	\centering
	\begin{lrbox}{\figbox}
		\begin{tikzpicture}[xscale=1.4,yscale=1.4]
			\draw[gray,fill=gray!30,even odd rule] (0,-1) -- (3,0) -- (2,2) -- (3,5) -- (-4,3) -- (-3,0) -- cycle
				(1,1) -- (1,3) -- (0,2) -- (0,1) -- cycle
				(-2,2) -- (-1,2) -- (-1,1) -- cycle;
			\draw (-1.25,1.75) node [inner sep = 1pt, circle, fill=black,label=below:{$\sigma_1$}] {}
			(0.5,2) node  [inner sep=1pt, circle, fill=black,label=below:{$\sigma_2$}] {};
			\draw [dashed] plot [smooth cycle,tension=1] coordinates {(-1.7,1) (-0.8,1) (-0.8,2.2) (-2.5,2)}
				plot [smooth cycle] coordinates {(-0.2,0.7) (-0.3,2.4) (1.1,3.2) (1.2,0.8)};
			\draw plot [smooth cycle] coordinates {(-2.3,0.7) (-2.7,2.2) (1.3,3.4) (1.2,0.5)}
			plot [smooth cycle,tension=1.5] coordinates {(-1,-0.1) (0,-0.6) (0.5,-0.2)}
				plot [smooth cycle,tension=1] coordinates {(2.2,3.8) (2.6,4.5) (2.0,4.4) (1.2,3.8)};
		\end{tikzpicture}
	\end{lrbox}
	\resizebox{0.7\textwidth}{!}{\usebox\figbox}
	\caption{Example of a region $\Pi$ and set $\Sigma=\left\{ \sigma_1,\sigma_2 \right\}$ used to specify monodromy factors; in this example the solid curves have monodromy 1 and the dashed curves have monodromy $-1$, independent of their orientation.}
	\label{fig:Pi_Sigma}
\end{figure}
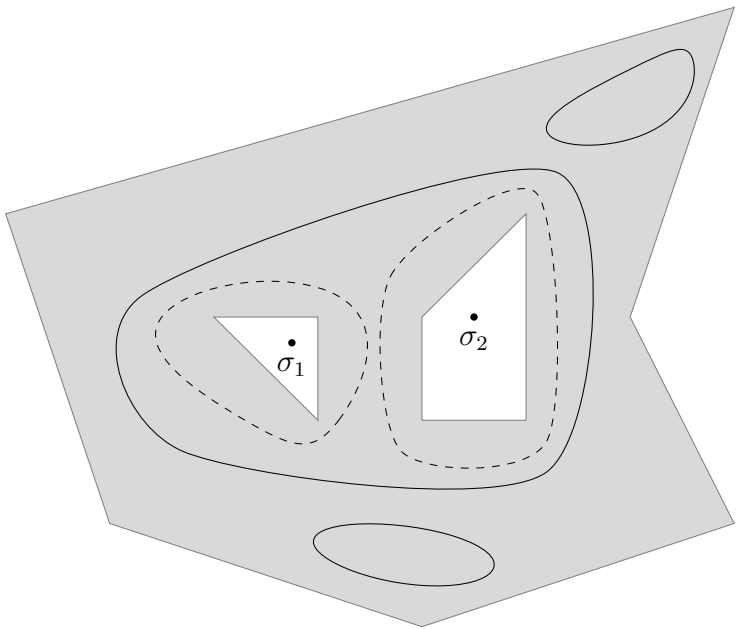

The main result of the present paper is the following:
\begin{theorem}
		\begin{equation}
			\log \detLS \wt\Delta_{L \Pi}
			= 
			\begin{aligned}[t]
				&
				\# \left((L \Pi) \cap \bZ^2 \right) \alpha_0
				+
				\sum_{e \in \fE(L\Pi)} \alpha_1(e) |e|
				+
				\sum_{c \in \fC(L\Pi)} \alpha_2(c)
				\\ & 
				+
				\alpha_3 (\Pi) \log L
				+
				\alpha_4(\Pi,\Sigma)
				+
				\cO\left( \frac{[\log L]^{23/14}}{L^{2/7}} \right)
			\end{aligned}
			\label{eq:main}
		\end{equation}
		as $L \to \infty$, where
		\begin{enumerate}
			\item $\#$ denotes the cardinality;
			\item $\fE(L\Pi)$ is the set of edges of $L\Pi$, $|e|$ is the length of the edge, and $\alpha_1(e)$ depends only in the slope of $e$ 
			\item $\fC(L\Pi)$ is the set of vertices (corners) of $L\Pi$, and $\alpha_2(c)$ depends only on the opening angle of the corner and its orientation relative to the coordinate axes, or in other words the slopes of the edges incident to $c$ and whether the associated interior angle is acute or obtuse;
			\item $\alpha_3 (\Pi)$ is the sum over corners of $\Pi$ of $(\pi^2 - \theta^2)/12\pi \theta$ where $\theta$ is the angle measured in the interior of the corner;
			\item $\alpha_4 (\Pi,\Sigma)$ is the logarithm of the $\zeta$-regularized determinant of the continuum Laplacian on a double cover of $\Pi$ branched around $\Sigma$ corresponding to the continuum limit of  $\wt \Delta_{L\Pi}$, defined in \cref{sec:zeta} below.
		\end{enumerate}
	\label{thm:main}
\end{theorem}
It is known from calculations using other techniques that $\alpha_0 = 4 G /\pi$, where $G$ is Catalan's constant; this can be verified from the expression in \cref{eq:alpha0_def} below, but the details are unedifying so I will not include them here.  As far as I know, $\alpha_1$, which is defined in \cref{eq:alpha1_def}, admits such an explicit expression only in special cases (for example it follows from \cite[Equation~(4.23)]{DD88} that $\alpha_1(e) = -\tfrac12 \log(\sqrt2 +1)$ when $e$ is parallel to one of the coordinate axes).

Note that \cref{eq:main} is written as a mixture of terms involving the geometry of $L\Pi \cap \bZ^2$ as a graph with a boundary and terms involving the limiting shape $\Pi$, rather than an expansion in orders of $L$.  %
In fact the first term on the right hand side of \cref{eq:main} is quadratic (by Pick's theorem), but with linear and constant terms which are generically nonzero and have no particular relationship to the other terms in the series, and $\sum_C \alpha_2(C)$ is independent of $L$.  As a result, while \cref{thm:main} implies that 
\begin{equation}
	\log \detLS \wt\Delta_{L \Pi}
	= 
	\begin{aligned}[t]
		&
		\alpha_0
		L^2
		+
		\alpha_1' (\Pi) L
		+
		\alpha_3 (\Pi) \log L
		+
		\alpha_4'(\Pi,\Sigma)
		+
		\cO\left( \frac{[\log L]^{23/14}}{L^{2/7}} \right)
		,
	\end{aligned}
	\label{eq:by_orders}
\end{equation}
the constant term in this expression is not generically given by the logarithm of the zeta-regularized determinant without further corrections.

\medskip

The order of the error term is mainly fixed by the estimates in \cref{sec:convergence_to_continuous,sec:discrete_ref} on the speed of convergence of the discrete heat kernel with Dirichlet boundary conditions to its continuum counterpart; the peculiar dependence is on $L$ is a result of combining a number of bounds which are helpful in different regimes and are clearly not themselves optimal.  More concretely, the dominant contribution to the error term is expressed in terms of the probability of a rescaled Brownian bridge started at a point $x \in \Pi$ coming within a small distance $\delta$ of the boundary of $\Pi$ without going further than $\delta$ outside of $\Pi$; a more careful analysis of this problem would presumably lead to an improvement of the estimate in \cref{eq:main} and perhaps a modest simplification of the proof. 
If the convergence of the heat kernel were as fast as in the full-plane case ($\cO(1/n^2$), the error term would be improved to
\begin{equation}
	\cO\left( L^2 \int_{L^2/(\log L)^\C}^\infty \frac{\dd t}{t^3} \right)
	=
	\cO\left( \frac{\Clast^2}{L^2} \right)
	\label{eq:optimistic_error}
\end{equation}
which is probably the highest precision which can be studied by the methods presented here.

\medskip
The $\zeta$-regularized determinant of the Dirichlet Laplacian is a very well-studied object, not least because it gives a way of defining a (finite) partition function for free quantum field theories and has a very accessible relationship to the geometry of the domain on which it is defined; it has consequently been used to study the Casimir effect \cite{DC.zeta}, the effects of space-time curvature \cite{Hawking}, and the definition of conformal field theories \cite{Gawedzki.lectures}.
In the case where $\Pi$ is simply connected (i.e.\ a polygon) it has been described in great detail in \cite{AS.polygon_zeta}.
Among other things, the fact that $\alpha_3$ is the logarithm of such a determinant implies that $\alpha_3 (L\Pi,L\Sigma) = \alpha_2 (\Pi) \log L + \alpha_3(\Pi,\Sigma)$ (see \cref{eq:zeta_rescaling}), so \cref{eq:main} is equivalent to
\begin{equation}
	\log \detLS \wt\Delta_{L \Pi}
	= 
	\begin{aligned}[t]
		&
		\# \left((L \Pi) \cap \bZ^2 \right) \alpha_0
		+
		\sum_{e \in \fE(L\Pi)} \alpha_1(e) |e|
		+
		\sum_{c \in \fC(L\Pi)} \alpha_2(c)
		\\ & 
		+
		\alpha_4(L \Pi,L \Sigma)
		+
		\cO\left( \frac{[\log L]^{23/14}}{L^{2/7}} \right)
		,
	\end{aligned}
	\label{eq:main_alt}
\end{equation}
which I in fact prove first.

The form of this expression is related to the fact that $\alpha_3$ and $\alpha_4$ can be defined in terms of the continuum Laplacian; on the other hand $\alpha_0$ and $\alpha_1$ are defined in terms of random walks on $\bZ^2$ and so should depend on the details of the graph and (for $\alpha_1$) the details of how the discrete boundary is chosen; $\alpha_2$ involves both discrete and continuous quantities.
The leading term $\alpha_0$, which is unambiguously defined and explicit, is known to correspond to the choice of $\bZ^2$ as the underlying infinite graph.
The boundary term is generically less explicit, but it is a sum of local terms which depend only on the local shape of the boundary.

One consequence of \cref{eq:main} is that since $\alpha_4$ is the first term which depends on $\Sigma$, for two different $\Sigma$ on the same $\Pi$ 
\begin{equation}
	\frac{\detsub{L \Sigma_1} \wt \Delta_{L \Pi}}{\detsub{L\Sigma_2} \wt \Delta_{L \Pi}}
	=
	\exp\left( \frac{\alpha_4 (\Pi,\Sigma_1)}{\alpha_4 (\Pi,\Sigma_2)}  \right)
	+o(1)
	, \quad
	L \to \infty
	;
	\label{eq:Z_ratio}
\end{equation}
a similar result was shown for a variety of toric graphs (not necessarily $\bZ^2$, or even periodic) in \cite{DubGhe}.

\medskip

Expansions to this order for determinants of Laplacians (or dimer partition functions, or Ising model partition functions, which are closely related under suitable conditions) have now been studied for several decades.  The presence of a term of constant order is interesting for applications to statistical physics, for example since it is related to the notion of central charge in conformal field theory \cite{Cardy.universal,Affleck.universal}.

When comparing these results it is important to note that boundary conditions (which, in the dimer model, are expressed in details of the geometry chosen) can have dramatic effects.
For example, the first publication containing asymptotic formulas of this type \cite{Ferdinand67} considered the partition function of the dimer model on a discrete rectangle, which is not a Temperleyan graph, and so does not correspond to Dirichlet boundary conditions for the Laplacian; consequently the expansion obtained there (which, in particular, has no logarithmic term) is not a special case of \cref{thm:main}, although it can be placed in a common framework.

The aforementioned publication by Ferdinand in 1967 studied the discrete torus as well as a rectangle, and noted the presence of a term of constant order expressed in terms of special functions.  This expansion started from the expression of the partition function as a Pfaffian (square root of the determinant) or linear combination of Pfaffians of matrices which had been explicitly diagonalized, and then carried out an asymptotic expansion using the Euler-Maclaurin formula.
Subsequently this approach was extended to other geometries, such as the cylinder and Klein bottle, and to obtain the asymptotic series to all orders in the domain size \cite{IOH03,BEP}.

The first work to consider a case included in \cref{thm:main} was
\cite[in particular Equation~(4.23)]{DD88}, who examined the determinant of the Laplacian on a rectangular lattice with Dirichlet boundary conditions (via the product of the nonzero eigenvalues of the Laplacian with Neumann boundary conditions on the dual lattice, adapting an earlier calculation \cite{Barber70}, which was the first to include a logarithmic term); this was the first work to explicitly identify part of the expansion of the discrete Laplacian determinant with a zeta-regularized determinant.
Subsequently, Kenyon \cite{Ken00} used the result of this calculation (which was based on an exact diagonalization of the Laplacian), combined with a formula using dimer correlation functions to relate partition functions on different domains, to obtain a result for rectilinear polygons, that is, simply connected domains with sides parallel to the axes of $\bZ^2$.  
Kenyon gives an alternative characterization of the term $\alpha_3$, related to the limiting average height profile of the associated dimer model; recently Finski \cite{Finski} showed that this is in fact the logarithm of the zeta-regularized determinant in a work which also generalized the result to, among other things, nonsimply-connected domains (including nontrivial connection), albeit still with the restriction that the boundary of the domain should always be parallel to one of the coordinate axes.

While revising this manuscript I became aware of a recent work by Izyurov and Khristoforov \cite{IK} who, using an approach which overlaps to a large extent with my own, obtain a comparable result including some cases which are excluded by \cref{thm:main} and excluding some of the cases I consider.  They consider a subset of an infinite periodic locally-planar graph which is invariant under reflections (not necessarily $\bZ^2$) and allow for Neumann or mixed Dirichlet-Neumann boundary conditions, as well as allowing the boundary to include punctures or slits.  However, they consider only boundary components lying on lines of symmetry of the underlying graph, which (since the graph is periodic) means that the corner angles are restricted to be multiples of $\pi/2$ or $\pi/3$.

In the last two decades there have also been a number of works 
relating the determinant of the Laplacian on higher-dimensional analogs of the discrete torus or rectangle \cite{CJK10,Vertman,HK20} to the corresponding zeta-regularized determinants.  
Although these authors introduced ideas which illuminate several aspects of the problem and make it possible to treat arbitrary dimension, they only treated graphs for which the spectrum of the Laplacian is given quite explicitly (although \cite{CJK10} does not use this directly, their treatment is based on the fact that the discrete torus is a Cayley graph, which is essentially what allows the spectrum to be calculated so explicitly).
For some self-similar fractals, the problem has been approached using recursive relationships which characterize the spectrum less explicitly \cite{CTT18}, but this leads us into a rather different context.

\medskip
As in \cite{CJK10}, the proof of \cref{thm:main} is based on the relationship between the (matrix or zeta-regularized) determinant and the trace of the heat kernel, reviewed in \cref{sec:zeta}; however, once this relationship is introduced, I (like the authors of \cite{IK}) treat the heat kernel using only a probabilistic representation, as the transition probability of a random walk or Brownian motion.
In \cref{sec:hk}, I introduce this representation (which is nonstandard only because of the presence of nontrivial monodromy factors) and present some straightforward bounds on the effect of changes in the domain $\Omega$ and the set $\Sigma$ which encodes the monodromy.
Such bounds can be used to obtain an expansion for the behaviour of the trace of the continuum heat kernel at small time (small, that is, on the scale associated with the domain); in \cref{sec:ref_geom}, I review the proof of this expansion given in \cite{Kac66}, which also provides the occasion to introduce certain definitions which will be used again later on to identify corresponding contributions in the discrete heat kernel.
For times which are large on the lattice scale, the trace of the discrete heat kernel converges to that of the corresponding continuum heat kernel; in \cref{sec:convergence_to_continuous} I give quantitative estimates on this convergence, based on the dyadic approximation.
Then in \cref{sec:discrete_ref} I use a combination of the methods of the previous sections to control the behavior of the lattice heat kernel in a regime corresponding to that of the expansion for its continuum counterpart considered in \cref{sec:ref_geom}.
Finally, in \cref{sec:conclusion} I combine all of these estimates to conclude the proof of \cref{thm:main}.

\medskip

\Cref{thm:main} is limited to (1) planar regions, (2) polygons whose corners have integer coordinates, (3) Dirichlet boundary conditions, and (4) subsets of the square lattice, rather than more general planar graphs; let me comment on these limitations.
\begin{enumerate}
	\item \label{it:shape}
		There does not seem to be any particular difficulty in extending the result to higher-dimensional polytopes, should such a result be of interest; however several of the motivations for studying the problem are particular to the two-dimensional case.  Similarly, it is quite straightforward to generalize this treatment to a torus, cylinder, Klein bottle, or M\"obius strip, but these cases are all already well understood by other means \cite{IOH03,BEP,KSW.tori,Cimasoni.Klein}.
		Domains with punctures or slits \cite{KW20} are presumably accessible with some technical improvements; in fact these have already been treated by a slightly different approach \cite{IK}.
	\item The restriction to this class of polygonal domains is related to the construction of the term involving the ``surface tension'' $\alpha_1$ in \cref{eq:main}.  
		Near the edges of the domain and away from the corners (the situation at the corners is more complicated, but in the end the same considerations come into play), the domain coincides with an infinite half-plane in regions of size proportional to $L$, so that the heat kernel there is very well approximated by the heat kernel on the intersection of $\bZ^2$ with that half plane.  The construction of $\alpha_1$, given in \cref{sec:discrete_ref}, takes advantage of this, and in particular the fact that, because the slope of the boundary of the half-plane is rational, the result is a (singly\nobreakdash-)periodic graph.  
	\item \label{it:bc} The discrete heat kernel in Dirichlet boundary conditions has a particularly convenient probabilistic representation (see \cref{sec:hk} for details): 
		for a trivial connection ($\Sigma = \emptyset$), the diagonal elements are simply the probabilities that a standard random walk on $\bZ^2$ returns to its starting point after a specified number of steps without ever having left the domain.
		This characterization has many helpful properties (most immediately, it is strictly monotone in the choice of domain).  
		For a general connection, it is the expectation value of a function which is the monodromy factor of the trajectory of the random walk if it returns to its starting point without leaving the domain and zero otherwise, which is obviously bounded in absolute value by the probability appearing in the case of a trivial connection.
		The continuum heat kernel has a similar representation in terms of standard two-dimensional Brownian motion.

		This is the key difference between the my approach and that of \cite{IK}.  They are able to treat Neumann and, more importantly (since there is a duality relating Neumann and Dirichlet boundary conditions) \emph{mixed} boundary conditions by coupling the random walk under consideration to one on an infinite graph by reflecting the walk when it encounters a boundary component with Neumann boundary conditions.  This requires, however, that the infinite graph be invariant under these reflections, which further limits the shapes of the domain which can be considered as mentioned above.
	\item The restriction to the square lattice is most relevant for the technique I use to obtain quantitative control on the approximation between the discrete and continuum heat kernels in \cref{sec:convergence_to_continuous}.
		I use an approach based on the dyadic coupling of Komlós, Major and Tusnády (KMT) \cite{KMT75}, which is essentially a one-dimensional technique, but which can be applied to random walks on $\bZ^d$ and a few other graphs where the coordinates of the walker evolve independently.
		In fact the KMT coupling can be used to obtain a similar control on the convergence of sums of I.I.D.\ sequences of vectors \cite{Einmahl}, and thus random walks on infinite periodic graphs, as in \cite{IK}.  This is consistent with the  similar expansions obtained in the last few years for periodic graphs embedded in the torus \cite{KSW.tori} or Klein bottle \cite{Cimasoni.Klein}, which have a similar structure, including a constant term which is universal (i.e.\ independent of many of the details of exactly which graph is chosen).

		It remains unclear whether a similar approach can be carried through for non-periodic graphs, which are needed to approximate many non-planar geometries.
		For graphs embedded in the torus, the asymptotic ratio corresponding to \cref{eq:Z_ratio} is known to be the same for all graphs on which the simple random walk converges weakly to Brownian motion (a weaker notion than the one used here) \cite{DubGhe}, raising the question of whether a similar result holds for other geometries. 
		\label{it:KMT}
\end{enumerate}

\section{Heat kernels, zeta functions, and determinants}
\label{sec:zeta}

For $M$ a symmetric $m \times m$ matrix with positive eigenvalues $0 < \mu_1 \le \mu_2 \le \dots, \mu_m$, 
we define an entire function by
\begin{equation}
	\zeta_M(s)
	:=
	\sum_{j=1}^m
	\mu_j^{-s}
	;
	\label{eq:zeta_M_def}
\end{equation}
this definition has the consequence that
\begin{equation}
	\zeta_M' (0) =
	-
	\sum_{j=1}^m
	\log \mu_j
	=
	- \log \det M
	\label{eq:zeta_det_log}
\end{equation}
or equivalently $\det M = \exp\left( - \zeta_M'(0) \right)$; 
to evaluate this, we note that for $\Re s > 0$
\begin{equation}
	\zeta_M (s)
	=
	\sum_{j=1}^m
	\frac{1}{\Gamma(s)}
	\int_0^\infty 
	t^{s-1}
	e^{- t \mu_j}
	\dd t
	=
	\frac{1}{\Gamma(s)}
	\int_0^\infty 
	t^{s-1}
	\Tr e^{-t M}
	\dd t
	\label{eq:zeta_M_integral}
\end{equation}
where $\Gamma(s) := \int_0^\infty t^{s-1} e^{-t} \dd t$,
and consequently
\begin{equation}
	\zeta_M'(s)
	=
	\frac{1}{\Gamma(s)}
	\int_0^\infty 
	t^{s-1}
	\log t
	\Tr e^{-t M}
	\dd t
	-
	\frac{\psi(s)}{\Gamma(s)}
	\int_0^\infty 
	t^{s-1}
	\Tr e^{-t M}
	\dd t
	,
	\label{eq:zeta_M_prime_integral}
\end{equation}
where $\psi(s) := \Gamma'(s)/\Gamma(s)$.
As $s\to 0^+$, the two integrals become divergent at $t = 0$, where $\Tr e^{-tM} \to \Tr I$; however these divergences cancel: we have
\begin{equation}
	\frac{1}{\Gamma(s)}
	\int_0^{e^{-\gamma}} t^{s-1} \left[ \log t - \psi(s) \right] \dd t
	=
	\frac{1}{\Gamma(s)}
	\left[ \frac{\left( -1-\gamma s \right)e^{-\gamma s}}{s^2} - \frac{\psi(s) e^{-\gamma s}}{s} \right]
	=
	\cO(s)
\end{equation}
for $s \to 0^+$, since $1/\Gamma(s) = \cO(s)$ and $s \psi(s) = -1 - \gamma(s) + \cO(s^2)$ \cite[eq.~5.7.1,~5.7.6]{DLMF}, where $\gamma$ is the Euler-Mascheroni constant; note that the choice of $e^{-\gamma}$ as the upper limit of integration gives the best possible cancellation, which is not strictly necessary but simplifies the resulting expressions.
Since $\Tr \left( e^{-tM} - I \right) = \cO(t)$ for $t \to 0$, the remaining parts of the integrals on the right hand side of \cref{eq:zeta_M_prime_integral} are convergent, leaving
\begin{equation}
	\zeta_M'(0)
	=
	\lim_{s \to 0^+}
	\zeta_M'(s)
	=
	\int_0^{e^{-\gamma}} 
	\Tr 
	\left( e^{-tM} - I \right) 
	\frac{\dd t}{t}
	+
	\int_{e^{-\gamma}}^\infty
	\Tr e^{-tM}
	\frac{\dd t}{t}
	.
	\label{eq:zeta_prime_M}
\end{equation}

In our case, we are therefore interested in $\Tr e^{-t \wt\Delta_{\Omega,\rho}}$, which is related to a continuous time random walk.  It will mostly be easier to consider the corresponding discrete-time random walk, which can be obtained through the following manipulations.
Noting
\begin{equation}
	\Tr e^{-t \wt\Delta_{\Omega,\rho}}
	=
	e^{-4 t} \Tr e^{t \left( 4 I - \wt\Delta_{\Omega,\rho} \right)}
	=
	\sum_{n=0}^\infty \frac{e^{-4 t }(4t)^n}{n!} 
	\Tr (I-\tfrac14 \wt\Delta_{\Omega,\rho})^n
	,
\end{equation}
where $I$ is the identity matrix on $\Omega \cap \bZ^2$,
\cref{eq:zeta_prime_M} becomes
\begin{equation}
	\begin{split}
		\zeta_{\wt\Delta_{\Omega,\rho}}' (0)
		&
		=
		\begin{aligned}[t]
			&
			\int_0^\infty \left( e^{-4 t} - \ind_{[0,e^{-\gamma}]}(t) \right) \frac{\dd t}{t}
			\Tr I
			\\ & \qquad
			+ \sum_{n=1}^\infty
			\frac{4}{n!}
			\int_0^\infty 
			e^{-4t} (4t)^{n-1}
			\dd t
			\Tr (I-\tfrac14 \wt\Delta_{\Omega,\rho})^n
		\end{aligned}
		\\ & = 
		- (\log 4) \# \left( \Omega \cap \bZ^2 \right) 
		+
		\sum_{n=1}^\infty
		\frac1{2n} 
		\Tr (I-\tfrac14 \wt\Delta_{\Omega,\rho})^{2n}
	\end{split}
	\label{eq:zeta_M_sum}
\end{equation}
(note that in the last step I have omitted the traces of odd powers, which vanish).

To obtain comparable continuum entities, let $\Pi_\Sigma$ denote a double cover of $\Pi$ branched around each point in $\Sigma$, and let $\Delta_{\Pi,\Sigma}$ be the Laplacian with Dirichlet boundary conditions acting on antisymmetric functions on $\Pi_\Sigma$, as detailed in \cref{app:DC}.
For consistency with \cref{eq:Lap} we use the positive Laplacian $\Delta = - \frac{\partial^2}{\partial x ^2} - \frac{\partial^2}{\partial y^2}$, giving a sign change and a factor of two with respect to many of the references cited.

The spectrum of $\Delta_{\Pi,\Sigma}$ is discrete and bounded from below, and its density of states is asymptotically approximately linear (Weyl's law\footnote{%
		More precisely, in two dimensions Weyl's law states that the number $N(\lambda)$ of eigenvalues of the Laplacian $\Delta_\Omega$ less than or equal to $\lambda$ (with multiplicity) satisfies
		\begin{equation}
			\nonumber
			\lim_{\lambda \to \infty}
			\frac{N(\lambda)}{\lambda}
			=
			\frac{|\Omega|}{2\pi}
			.
		\end{equation}
		It is usually stated for the Laplacian on functions on planar domains, but the well-known variational proof \cite[Section~6.4]{CourantHilbert1} also applies to functions on double covers.
}),
so the zeta function
\begin{equation}
	\zeta_{\Delta_{\Pi,\Sigma}}(s)
	:=
	\sum_{\lambda \in \spec \Delta_{\Pi,\Sigma}}
	\lambda^{-s}
	\label{eq:zeta_Delta_def}
\end{equation}
is well defined and analytic for $\Re s >1$; as we shall shortly see, $\zeta_{\Delta_{\Pi,\Sigma}}$ can be analytically continued to a neighborhood of 0, and -- by analogy with \cref{eq:zeta_det_log} --  $e^{-\zeta_{\Delta_{\Pi,\Sigma}}'(0)}$ is called the \emph{$\zeta$-regularized determinant} of the Laplacian.

The analogues of \cref{eq:zeta_M_integral,eq:zeta_M_prime_integral} hold for $\Re s > 1$ (as we shall see in \cref{thm:DC_trace} below, $e^{-t \Delta_{\Pi,\Sigma}}$ is trace class for $t > 0$); the trace appearing in the integral has has a noteworthy asymptotic expansion~\cite{BB62}
\begin{equation}
	\Tr e^{-t \Delta_{\Pi,\Sigma}}
	=
	a_0(\Pi) t^{-1}
	+
	a_1(\Pi) t^{-1/2}
	+
	a_2(\Pi)
	+
	\cO(e^{-\kappa^2(\Pi)/t})
	, \quad t \to 0 +
	\label{eq:ht_asympt}
\end{equation}
where $a_0(\Pi)$ is $1/4\pi$ times the area of $\Pi$, $a_1(\Pi)$ is $-1/8\sqrt{\pi}$ times the length of the boundary, and $a_2(\Pi)$ is the sum over corners of $\left( \pi^2 - \theta^2 \right)/24\pi\theta$ where $\theta$ is the opening angle measured from the interior of $\Pi$\footnote{The first explicit formula for $a_2$ appeared in \cite{Kac66}; this simpler version, due to {D.\ B.\ Ray}, was first published in \cite{MS67.eig}.
}; an expansion for smooth manifolds with smooth boundaries (or no boundaries) is also known \cite{MS67.eig,Gilkey}.
Although \cite{BB62,Kac66,MS67.eig} were formulated for the plane Laplacian or Laplace-Beltrami operator (equivalent to $\Sigma = \emptyset$), the proof of \cref{eq:ht_asympt} in \cite{Kac66} extends straightforwardly to the dual cover.  I will present a version of this proof in \cref{sec:ref_geom}, both for completeness and because some of the entities introduced there are helpful for comparison with the discrete case.

$\Tr e^{-t \Delta_{\Pi,\Sigma}}$ also decreases exponentially as $t \to \infty$, since the spectrum of $\Delta_{\Pi,\Sigma}$ is bounded away from 0 (see \cref{app:DC}, especially \cref{eq:D_domain_monotone}), 
and so letting 
\begin{equation}
	F_{\Pi,\Sigma}(t) 
	:=  
	\Tr e^{-t \Delta_{\Pi,\Sigma}} 
	-   
	a_0(\Pi) t^{-1}
	-
	a_1(\Pi) t^{-1/2}
	-
	a_2(\Pi) \ind_{[0,e^{-\gamma}]}(t)
	\label{eq:F_Pi_def}
\end{equation}
we obtain, for any $K > e^{-\gamma}$ and $\Re s > 1$,
\begin{equation}
	\begin{split}
		\int_0^K
		t^{s-1}
		&
		\Tr e^{-t \Delta_{\Pi,\Sigma}} 
		\dd t
		-
		\int_0^K
		t^{s-1}
		F_{\Pi,\Sigma}(t)
		\dd t
		\\ = &
		\int_0^K
		\left( a_0(\Pi) t^{s-2} + a_1(\Pi) t^{s-3/2} \right)
		\dd t
		+
		a_2(\Pi) \int_0^{e^{-\gamma}} t^{s-1} \dd t
		\\ = &
		a_0(\Pi)
		\frac{K^{s-1}}{1-s}
		+
		a_1(\Pi)
		\frac{K^{s-1/2}}{\tfrac12 - s}
		+
		\frac{a_2(\Pi)}{s} e^{-\gamma s}
		;
	\end{split}
\end{equation}
we thus have 
\begin{equation}
	\begin{split}
		\Gamma(s) 
		\zeta_{\Delta_{\Pi,\Sigma}} (s)
		=&
		\int_K^\infty
		t^{s-1}
		\Tr e^{-t \Delta_{\Pi,\Sigma}} 
		\dd t
		+
		\int_0^K
		t^{s-1}
		F_{\Pi,\Sigma}(t)
		\dd t
		+
		a_0(\Pi)
		\frac{K^{s-1}}{1-s}
		\\ &
		+
		a_1(\Pi)
		\frac{K^{s-1/2}}{\tfrac12 - s}
		+
		\frac{a_2(\Pi)}{s}
		e^{-\gamma s}
		,
	\end{split}
\end{equation}
and we can analytically continue and then take the limit $K \to \infty$ to obtain
\begin{equation}
	\zeta_{\Delta_{\Pi,\Sigma}} (s)
	=
	\frac{1}{\Gamma(s)}
	\left[ 
		\int_0^\infty
		t^{s-1}
		F_{\Pi,\Sigma}(t)
		\dd t
		+
		\frac{a_2(\Pi)}{s}
		e^{-\gamma s}
	\right]
\end{equation}
for all $|s| < \tfrac12$, and noting that
\begin{equation}
	\frac{\dd}{\dd s}
	\left.
		\int_0^\infty t^{s-1}
		F_{\Pi,\Sigma}(t)
		\dd t
	\right|_{s=0}
	=
	\int_0^\infty
	F_{\Pi,\Sigma}(t) 
	\log t
	\frac{\dd t }{t}
	<
	\infty
\end{equation}
along with
\begin{equation}
	\lim_{s \to 0}
	\frac{1}{\Gamma(s)}
	=
	0
	, \ 
	\lim_{s \to 0}
	\frac{\dd}{\dd s}
	\frac{1}{\Gamma(s)}
	=
	1
	, \ 
	\lim_{s \to 0}
	\frac{1}{s \Gamma(s)}
	=
	1
	, \textup{ and }
	\lim_{s \to 0}
	\frac{\dd}{\dd s}
	\frac{1}{s \Gamma(s)}
	=
	\gamma
\end{equation}
gives
\begin{equation}
	\zeta'_{\Delta_{\Pi,\Sigma}}(0)
	=
	\int_0^\infty F_{\Pi,\Sigma}(t) \frac{\dd t}{t}
	.
	\label{eq:z_reg_Delta_final}
\end{equation}

Note, finally, that from the definitions of $F$ and the various $a_n(\Pi)$ we have 
\begin{equation}
	\begin{split}
		F_{L \Pi,L\Sigma} (t)
		= &
		\Tr e^{- (t/L^2) \Delta_{\Pi,\Sigma}}
		-
		a_0(\Pi)
		\frac{L^2}{t}
		-
		a_1(\Pi)
		\frac{L}{t^{1/2}}
		-
		a_2(\Pi)
		\ind_{[0,e^{-\gamma} L^2]}(t)
		\\ = &
		F_{\Pi,\Sigma} \left( t/L^2 \right)
		+
		a_2(\Pi)
		\ind_{(1,L^2]}(e^\gamma t)
	\end{split}
	\label{eq:F_rescale}
\end{equation}
whence
\begin{equation}
	\zeta'_{\Delta_{L\Pi,L\Sigma}}(0)
	=
	\zeta'_{\Delta_{\Pi,\Sigma}}(0)
	+
	2 a_2(\Pi) \log L
	.
	\label{eq:zeta_rescaling}
\end{equation}

\section{Heat kernels, Brownian motion, and random walks}
\label{sec:hk}

I now turn to the representation I will use to estimate the trace of $(I-\tfrac14 \wt\Delta_{\Omega,\rho})^n$ appearing in \cref{eq:zeta_M_sum}.  
When $\rho \equiv 1$, the elements of this matrix are simply the transition probabilities for $n$ steps of a random walk which is killed on leaving $\Omega$; to be precise
\begin{equation}
	\left[  (I-\tfrac14 \wt\Delta_{\Omega,\rho})^n\right]_{xy}
	=
	\wt P_\Omega (x,y;n)
	:=
	\bP_x \left[ \wt \cW_n = y, \ \wt T_\Omega > n  \right]
	\label{eq:PDt_Omega_base}
\end{equation}
where $\wt \cW$ is a simple random walk on $\bZ^2$ starting at $x$,
and $\wt T_X$, $X \subset \bR^2$, denotes the first time at which $\wt \cW$ leaves $X$, which in light of the definition \cref{eq:Lap}, we understand to be the first time it jumps along a bond associated with a line segment not entirely contained in $X$.
For general $\rho$, 
\begin{equation}
	\left[  (I-\tfrac14 \wt\Delta_{\Omega,\rho})^n\right]_{xy}
	=
	\bE_x \left[ 
		\ind\left\{ \wt W_n = y \right\}
		\ind\left\{\wt T_\Omega > n \right\}
		\prod_{j=1}^n
		\rho_{\wt \cW_{j-1},\wt \cW_j}
	\right]
	.
\end{equation}
Letting $W_z (\wt \cW;n)$ denote the winding number (index) of $\wt \cW$ around $z$ up to time $n$, when $\rho$ is associated with some $\Sigma$, for the diagonal elements (which, ultimately, are the ones we are interested in) are given by
\begin{equation}
	\left[  (I-\tfrac14 \wt\Delta_{\Omega,\rho})^n\right]_{xx}
	=
	\wt P_{\Omega,\Sigma} (x,x;n)
	:=
	\bE_x\left[ 
		\ind\left\{ \wt W_n = y \right\}
		\ind\left\{\wt T_\Omega > n \right\}
		e^{i \pi W_\Sigma (\wt \cW;n)}
	\right]
	\label{eq:PDt_Omega_Sigma}
\end{equation}
with $W_\Sigma (\wt \cW; n) := \sum_{\sigma \in \Sigma} W_\sigma (\wt \cW; n)$.

Let us now look at the continuum.  With $\Sigma = \emptyset$, $\Delta_{\Omega,\Sigma}$ is just the usual Dirichlet Laplacian $\Delta_\Omega$ on $\Omega$, and so for $t>0$ $e^{-t \Delta_\Omega}$ is the semigroup giving the solutions of the heat equation with Dirichlet boundary conditions.
The Feynman-Kac formula then gives, for any suitable $f$, 
\begin{equation}
	\begin{split}
		\left[ e^{-t \Delta_\Omega} f \right](x)
		&=
		\bE_x\left[ \ind_{(2 t,\infty)}(T_\Omega) f(\cW_{2t}) \right]
		\\ & =
		\bE_x\left[ f(\cW_{2t}) \right]
		-
		\bE_x\left[ \ind_{[0,2t]} (T_\Omega) \bE_{\cW_{T_\Omega}}\left[ f\left( \cW_{t-T_\Omega} \right) \right] \right]
		\\ & = 
		\int_\Omega
		\left\{ 
			P(x,y; 2t) - \bE_x\left[ \ind_{[0,2t]} (T_\Omega) P(\cW_{T_\Omega},y;2t-T_\Omega) \right]
		\right\} 
		f(y)
		\dd y
	\end{split}
\end{equation}
where $\cW$ is a two-dimensional Brownian motion starting from $x$, $T_\Omega$ is its first exit time from $\Omega$, and 
\begin{equation}
	P(x,y;t) 
	=
	\frac{1}{2 \pi t} e^{-|x-y|^2/2 t}
	\label{eq:P_base}
\end{equation}
is the full-plane heat kernel;
in other words $e^{-\tfrac{t}{2}\Delta_\Omega}$ has integral kernel
\begin{equation}
	P_\Omega (x,y;t)
	:=
	P(x,y;t) - \bE_x\left[ \ind_{[0,t]} (T_\Omega) P(\cW_{T_\Omega},y;t-T_\Omega) \right]
	.
	\label{eq:PD_Om_def}
\end{equation}
Apart from the boundary (i.e.\ for $y \in \Omega$) it is evident from this expression that $P_\Omega$ is a smooth function of $y$ for $t > 0$, 
and so %
\begin{equation}
	P_\Omega (x,y;t)
	=
	\lim_{r \to 0^+}
	\frac{1}{\pi r^2}
	\bP_x \left[ \left|y - \cW_t \right| \le r, \ T_\Omega > t \right]
	,
	\label{eq:PD_Omega_base}
\end{equation}
which is helpful in deriving estimates on $P_\Omega$.

The Laplacian acting on functions on a double cover with prescribed monodromy gives a counterpart to \cref{eq:PDt_Omega_base}.  Although the following statement is easily obtained by standard methods, it does not appear to be a ready-made corollary of any results I was able to find, so I provide a proof in \cref{app:DC}:
\begin{theorem}
	\begin{equation}
		\Tr e^{- t \Delta_{\Pi,\Sigma}}
		=
		\int_{\Pi} P_{\Pi,\Sigma} (x,x;2 t) \dd x
		\label{eq:trace_PD}
	\end{equation}
	(note that the integral is over $\Pi$, not the double cover $\Pi_\Sigma$)
	where 
	\begin{equation}
		P_{\Pi,\Sigma} (x,x;t)
		=
		\lim_{r \to 0^+}
		\frac{1}{\pi r^2}
		\bE_x\left[ 
			\ind\left\{ |x - \cW_t| \le r \right\}
			\ind\left\{ T_\Pi > t \right\}
			e^{i \pi W_\Sigma(\cW;t)}
		\right]
		\label{eq:PD_Sig}
	\end{equation}
	and $W_\Sigma(\cW;t)$ denotes the sum over $y \in \Sigma$ of the winding number of $\cW$ around $y$ in the time interval $[0,t]$, rounded to the nearest integer.
	\label{thm:DC_trace}
\end{theorem}
Note that the same limit is obtained without rounding off the winding number\footnote{Something of the sort is needed to sensibly define values of $P_{\Pi,\Sigma} (x,y;t)$ with $x \neq y$, but these values are not important in the context of this article.}, but the rounded version is more convenient for the calculations I will present below.

\medskip

This probabilistic representation can be used to prove a number of bounds.  Let us begin with the following estimate on the effect of changes in $\Omega$ and/or $\Sigma$, which will be one of the basic ingredients for the rest of the present paper.
\begin{theorem}
	For all $\Omega,\Theta \subset \bR^2$, $\Sigma_1 \subset \Omega^\comp$, $\Sigma_2 \subset \Theta^\comp$,
	$x \in \Omega \cap \Theta$,
	\begin{equation}
		\left| P_{\Omega,\Sigma_1}(x,x;t) - P_{\Theta,\Sigma_2}(x,x;t)\right|
		\le
		\frac{4}{\pi t}
		\exp\left( -[\dist(x,(\Theta \triangle \Omega) \cup (\Sigma_1 \triangle \Sigma_2))]^2/t \right)
		.
		\label{eq:PD_domain_change}
	\end{equation}
	\label{thm:PD_domain_change}
\end{theorem}

\begin{proof}
	For brevity let $R := \dist(x,(\Theta \triangle \Omega) \cup (\Sigma_1 \triangle \Sigma_2))$.
	Using \cref{eq:PD_Sig},
	\begin{equation}
		\begin{split}
			&
			P_{\Omega,\Sigma_1}(x,x;t) - P_{\Theta,\Sigma_2}(x,x;t)
			\\& \quad =
			\lim_{r \searrow 0}
			\frac{1}{\pi r^2}
			\bE_x \left[ 
				\ind_{[0,r]} (|\cW_t-x|) 
				\left\{ 
					\ind_{(t,\infty)}(T_\Omega) 
					e^{i \pi W_{\Sigma_1}(\cW; t)} 
					-
					\ind_{(t,\infty)}(T_\Theta) 
					e^{i \pi W_{\Sigma_2}(\cW; t)} 
				\right\}
			\right]
			.
		\end{split}
		\label{eq:PD_domain_change_1}
	\end{equation}
	For the difference in the expectation to be nonzero, $\cW$ must either enter $\Theta \triangle \Omega$ or wind around a point in $\Sigma_1 \triangle \Sigma_2$ before time $t$, 
	either of which requires it to travel a distance at least $R$ from its starting point; thus
	\begin{equation}
		\begin{split}
			\left| P_{\Omega,\Sigma_1}(x,x;t) - P_{\Theta,\Sigma_2}(x,x;t)\right|
			& \le
			2
			\lim_{r \searrow 0}
			\frac{1}{\pi r^2}
			\bP_x\left[ 
				|\cW_t - x| \le r
				, \ 
				T_{B_R(x)} \le t
			\right]
			,
		\end{split}
		\label{eq:PD_domain_change_2}
	\end{equation}
	where $B_R(x)$ is the circle of radius $R$ centered at $x$.
	The event of $\cW$ leaving $B_R(x)$ is contained in union of the events of moving by more than $\tfrac{\sqrt 2}{2} R$ in any one of the four cardinal directions, and the contribution of each of these events can be calculated using the reflection principle:
	\begin{equation}
		\begin{split}
			&
			\bP_x\left[ 
				|\cW_t - x| \le r
				, \ 
				T_{B_R(x)} \le t
			\right]
			\\ & \quad
			\le
			\sum_{\hat e = \pm \hat e_1, \hat e_2}
			\bP_x\left[ 
				|\cW_t - x| \le r
				, \ 
				\max_{s \in [0,t]} (\cW_s - x) \cdot \hat e \ge \tfrac{\sqrt 2}2 R
			\right]
			\\ & \quad
			=
			4 \int_{B_r (x + \sqrt 2 R \hat e_1)}
			P(x,y;t) \dd y
			,
		\end{split}
		\label{eq:PD_domain_change_3}
	\end{equation}
	and inserting this in \cref{eq:PD_domain_change_2}, together with the exact expression for $P$, gives \cref{eq:PD_domain_change}.
\end{proof}

For the lattice case it is also possible to bound the effect of changes in the geometry in similar fashion:
\begin{theorem}
	There is a constant $\Cl{dom_diff_disc} >0$ such that
	for all $\Omega,\Theta \subset \bR^2$, $\Sigma_1 \subset \Omega^\comp$, $\Sigma_2 \subset \Theta^\comp$,
	$x \in \Omega \cap \Theta \cap \bZ^2$,
	\begin{equation}
		\left|\wt P_{\Omega,\Sigma_1}(x,x;n) - \wt P_{\Theta,\Sigma_2}(x,x;n)\right|
		\le
		\frac{\Cr{dom_diff_disc}}{n}
		\exp\left( -[\dist(x,(\Theta \triangle \Omega) \cup (\Sigma_1 \triangle \Sigma_2))]^2/n \right)
		.
	\end{equation}
	\label{thm:PDt_domain_change}
\end{theorem}

\begin{proof}
	Recalling \cref{eq:PDt_Omega_Sigma},
	\begin{equation}
		\begin{split}
			& \wt P_{\Omega,\Sigma_1}(x,x;t) 
			- 
			\wt P_{\Theta,\Sigma_2}(x,x;t)
			\\ & \quad
			=
			\bE_x\left[ 
				\ind_{ \left\{ x \right\}} \left( \wt \cW_t \right)
				\left\{ 
					\ind_{(t,\infty)}\left( \wt T_{\Omega} \right)
					e^{i \pi W_{\Sigma_1} (\wt \cW;t)}
					-
					\ind_{(t,\infty)}\left( \wt T_{\Theta} \right)
					e^{i \pi W_{\Sigma_2} (\wt \cW;t)}
				\right\}
			\right]
			,
		\end{split}
		\label{eq:PDt_domain_change_1}
	\end{equation}
	and as in the proof of \cref{thm:PD_domain_change} this can be bounded as
	\begin{equation}
		\left|
			\wt P_{\Omega,\Sigma_1}(x,x;t) 
			- 
			\wt P_{\Theta,\Sigma_2}(x,x;t)
		\right|
		8 \bP_x\left[ \wt \cW_n = x + 2 R' \hat e_1 \right]
		,
	\end{equation}
	with $R'$ the smallest integer such that $R' \ge \tfrac{\sqrt 2}{2} R$.
	Noting that $\wt \cW_n - \wt \cW_0$ has the same distribution as $\tfrac{\hat e_1 + \hat e_2}2 W_n^+ + \tfrac{\hat e_1 - \hat e_2}2 W_n^-$ with $W_n^{\pm}$ two independent one-dimensional random walks started at 0,
	\begin{equation}
		\bP_x\left[ \wt \cW_n = x + 2 R' \hat e_1 \right]
		=
		\left[ 2^{-n} \binom{n}{R'+n/2} \right]^2
		,
		\label{eq:independent_RW_trick}
	\end{equation}
	which implies the announced result together with \cref{lem:RW_Gaussian}.
\end{proof}

\section{Reference geometries for the continuum heat kernel}
\label{sec:ref_geom}
In this section we review the proof of \cref{eq:ht_asympt}.  This is substantially the same proof as \cite{Kac66}, who considered the case of a planar region (i.e.\ without branching).  We present it here in order to make it clear that it works unchanged for the more general case and to bring out a number of details in a way which will make it easier to compare with the discrete version.

For any $r >0$ and $x \in \bR^2$, let $B_r(x)$ denote the open disk of radius $r$ centered at $x$.
Let $\kappa=\kappa(\Pi)>0$ be the such that, 
for all $x \in \Pi$, $B_\kappa(x) \cap \partial\Pi$ is contained in the union of two adjacent line segments of $\partial\Pi$. 
\begin{figure}[t]
	\centering
	\begin{tabular}{m{0.7\textwidth} m{0.2\textwidth}}
		\begin{lrbox}{\figbox}
			\tikzpicturedependsonfile{PI_R.tikz}
			\input{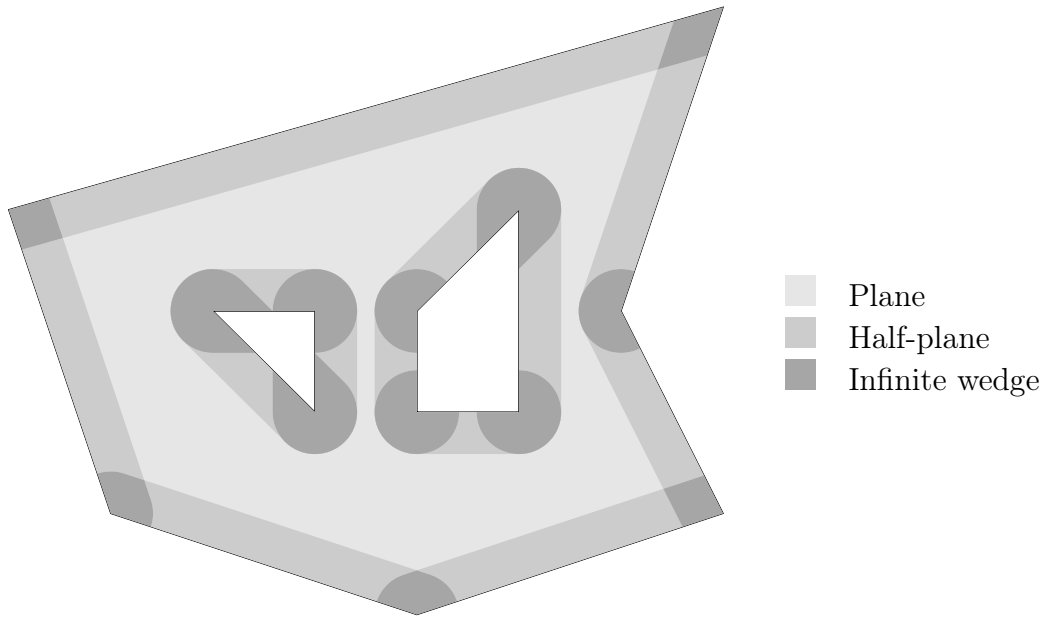}
		\end{lrbox}
		\resizebox{0.7\textwidth}{!}{\usebox\figbox}
		&
		\begin{tabular}{ll}
			\tikz[scale=0.4]{ \fill[gray!20] (0,0) rectangle (1,1);} 
			&
			Plane
			\\
			\tikz[scale=0.4]{ \fill[gray!40] (0,0) rectangle (1,1);} 
			&
			Half-plane
			\\
			\tikz[scale=0.4]{ \fill[gray!70] (0,0) rectangle (1,1);} 
			&
			Infinite wedge
		\end{tabular}
	\end{tabular}
		\caption{Example of a polygonal domain $\Pi$ showing the regions for which $R_\Pi(x)$ is of the specified form.}
	\label{fig:R_examples}
\end{figure}
We will approximate $P_\Pi (x,x;t)$ (and, in due time, $\wt P_{L\Pi} (x,x;t)$) by the heat kernel in a reference geometry, a set $R_\Pi (x)$ obtained by taking $B_{\kappa}(x) \cap \Pi$ and extending infinitely any components of $\partial\Pi$ present (see \cref{fig:R_examples}).
Note that this gives either the whole plane, a half-plane, or an infinite wedge, and that (up to translations) the same finite collection of wedges and half-planes appear for all $L \in \bN$.

When $R(x) = \bR^2$, we have $P_{R(x)}(x,x;t) = P_0(t) = 1/2\pi t$; note that
\begin{equation}
	\int_\Omega P_0(2 t) \dd t
	=
	\frac{|\Omega|}{4 \pi t}
	=
	a_0(\Omega) t^{-1}
	.
	\label{eq:ref_to_a0}
\end{equation}

\begin{figure}[t]
	\centering
	\begin{tabular}{lr}
		\begin{lrbox}{\figbox}
		\begin{tikzpicture}
			\fill[gray!30,even odd rule] (0,-1) -- (3,0) -- (2,2) -- (3,5) -- (-4,3) -- (-3,0) -- cycle
				(1,1) -- (1,3) -- (0,2) -- (0,1) -- cycle
				(-2,2) -- (-1,2) -- (-1,1) -- cycle;
			\draw (-0.5,3) node {$\Pi$};
			\draw  (0,-1) node[circle,draw,fill=white,inner sep=2pt,outer sep=0pt] {} -- + (-3,1) node[circle,fill=black,inner sep=2pt] {} node[midway, label=below:$\ell$] {};
		\end{tikzpicture}
		\end{lrbox}
		\resizebox{0.5\textwidth}{!}{\usebox\figbox}
		&
		\begin{lrbox}{\figbox}
		\begin{tikzpicture}
			\fill[gray!20] (-1.5,0.5) to [out=75,in=160] (1.5,4.5) -- (4.5,3.5) to [out=-20,in=75] (4.5,-1.5) -- cycle;
			\fill[gray!50] (3,-1) -- (0,0) -- (1.5,4.5) -- + (3,-1) -- cycle;
			\fill[gray] (3,-1) -- (0,0) -- (0.5,1.5) -- + (3,-1) -- cycle;
			\draw (1.5,-0.5) + (1,3) node {$E$};
			\draw (1.5,-0.5) + (0.2,0.6) node {$E_\kappa$};
			\draw  (0,0) -- (1.5,4.5);
			\draw [dashed] (3,-1) -- ++ (1.5,4.5);
			\draw [dashed] (-1.5,0.5) -- (0,0) (3,-1) -- (4.5,-1.5);
			\draw [dashed] (3,-1) %
			-- (0,0) %
			node[midway, label=below:{$\ell(E) = \ell$}] {};
			\draw (-0.3,2) node {$H(E)$};
		\end{tikzpicture}
		\end{lrbox}
		\resizebox{0.5\textwidth}{!}{\usebox\figbox}
	\end{tabular}
	\caption{Constructions related to the edge term.  On the left, an example of a polygonal region $\Pi$ with one of the half-open line segments $\ell$ of the boundary highlighted.
	On the right, (part of) the corresponding semi-infinite rectangle $E$ and half plane $H(E)$; note $E_\kappa \subset E \subset H(E)$.}
	\label{fig:edge}
\end{figure}
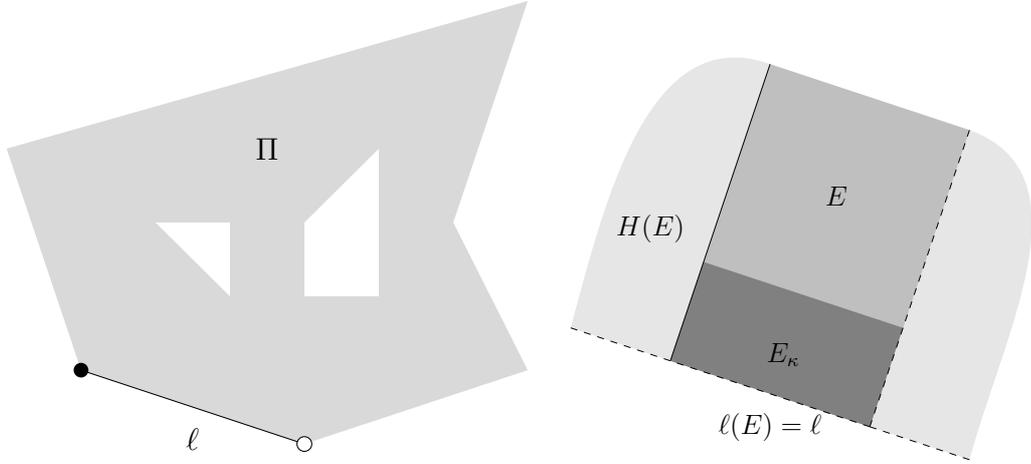
The boundary of $\Pi$ is a disjoint union of semi-open line segments.
Let $\cE$ be set whose elements are the semi-infinite rectangles swept out by taking one such line segment and displacing it perpendicularly in the direction of the interior of $\Pi$ and removing the original line segment, so that each $E \in \cE$ is closed on exactly one of its three sides\footnote{This is not immediately important, but when considering the discrete version it will be helpful that this has the consequence that $LE$, for $L \in \bN$, is a disjoint union of $L$ translates of $E$. \label{fn:LE}}.
For $E \in \cE$, let $\ell(E)$ denote the line segment used to construct $E$, $w(E)$ denote the width of $E$ (that is, the length of $\ell(E)$), so that $\sum_{E \in \cE} w(E)$ is the perimeter of $\Pi$; let $E_\kappa \subset E$ be the $w(E) \times \kappa$ rectangle with side $\ell(E)$.  Finally, let $H(E)$ be the open half plane which contains $E$ and whose boundary contains $\ell(E)$;
then for a given $E$ the set of $x\in \Pi$
with $R(x) = H(E)$ is a proper subset of $E_\kappa$.  Some aspects of these definitions are illustrated in Figure~\ref{fig:edge}.

Letting $x_\perp$ denote the component of $x$ perpendicular to $\ell(E)$, we have
\begin{equation}
	P_{H(E)}(x,x;t)
	=
	P_0(t)
	-
	\frac{1}{2 \pi t}
	e^{-x_\perp^2/t}
	=:
	P_0(t)
	+
	\partial P_{H(E)} (x;t)
	\label{eq:6P_def}
\end{equation}
thus
\begin{equation}
	\int_E 
	\partial P_{H(E)} (x;2 t)
	\dd x
	=
	-
	w(E)
	\frac{1}{8 \sqrt{\pi}}
	t^{-1/2}
\end{equation}
and thus
\begin{equation}
	\sum_{E \in \cE}
	\int_E 
	\partial P_{H(E)} (x,x;2 t)
	\dd x
	=
	a_1(\Pi) t^{-1/2}
	\label{eq:ref_to_a1}
\end{equation}
with $a_1$ as defined after \cref{eq:ht_asympt}.
Also,
\begin{equation}
	\begin{split}
		\left|
		\int_{E_\kappa} 
		\partial P_{H(E)} (x,x;2 t)
		\dd x
		-
		\int_E 
		\partial P_{H(E)} (x,x;2t)
		\dd x
		\right|
		= &
		w(E)
		\frac{1}{8 \sqrt{\pi t}}
		\erfc \left( \kappa/\sqrt{t} \right)
		\\ \le &
		\frac{w(E)}{8 \pi \kappa}
		e^{-\kappa^2/t}
		.
	\end{split}
\end{equation}

Finally, let $\cC$ be the set of $C \subset \bR^2$ which are open infinite wedges matching the corners of $\Pi$.  For such a $C$ let $\varphi(C)$ be its opening angle, and let $E^1(C),E^2(C)$ be the elements of $\cE$ associated with the line segments incident on the corner; 
thus $C = H(E^1(C)) \cap H(E^2(C))$ if $\varphi(C) \le \pi$ 
and $C = H(E^1(C)) \cup H(E^2(C))$ if $\varphi(C) > \pi$.
Also let $v(C)$ be the vertex of $C$, and let
$C_\kappa$ 
be the set of $x \in C$ which are within a distance $\kappa$ of both edges of $C$, or equivalently the set of $x \in \Pi$ for which $R_\Pi (x) = C$.
For $j=1,2$ let 
$\hat C^j_\kappa = (E^j)_\kappa(C) \setminus C$ and let $\hat C^j$ be the corresponding infinite wedge; both of these are $\emptyset$ if $\varphi(C) \ge \pi/2$.   
\begin{figure}[t]
	\centering
	\begin{tabular}{lr}
		\begin{lrbox}{\figbox}
			\begin{tikzpicture} [thick,record path/.style={/utils/exec=\tikzset{bent brace/.cd,#1},
			    decorate,decoration={markings,mark=at position 0 with
			    {\setcounter{bracep}{1}%
			    \path(0pt,{\pgfkeysvalueof{/tikz/bent brace/scale}*\pgfkeysvalueof{/tikz/bent brace/raise}})
			     coordinate (bracep-\number\value{bracep});
			     \pgfmathsetmacro{\mystep}{(\pgfdecoratedpathlength-4pt)/int(1+(\pgfdecoratedpathlength-4pt)/2pt)}
			     \xdef\mystep{\mystep}},
			    mark=between positions 2pt and {\pgfdecoratedpathlength-2pt} step \mystep pt with {\stepcounter{bracep}
			    \pgfmathtruncatemacro{\itest}{
			    ifthenelse(abs(\number\value{bracep}*2pt/\pgfdecoratedpathlength-0.5)<1pt/\pgfdecoratedpathlength,1,0)}
			    \coordinate (bracep-\number\value{bracep}) at
			    (0,{\pgfkeysvalueof{/tikz/bent brace/scale}*(\pgfkeysvalueof{/tikz/bent brace/raise}+2pt+\itest*2pt)})
			    \ifnum\itest=1
			     node[transform shape,bent brace/nodes]{\pgfkeysvalueof{/tikz/bent brace/text}}
			    \fi;},
			    mark=at position 1 with {\stepcounter{bracep}
			    \coordinate (bracep-\number\value{bracep}) at (0,\pgfkeysvalueof{/tikz/bent brace/scale}*\pgfkeysvalueof{/tikz/bent brace/raise});}
			    }},brace/.style={insert path={plot[smooth,samples at={1,...,\number\value{bracep}},variable=\x]
			    (bracep-\x)}},
			    bent brace/.cd,raise/.initial=0pt,scale/.initial=1,text/.initial={},
			    nodes/.style={},node options/.code=\tikzset{bent brace/nodes/.append style=#1},
			    mirror/.code={\tikzset{bent brace/scale=-1}}]
			\tikzmath{\radius = 3.0; 
				\r1 = \radius*0.3; %
				\r2 = \radius * 1.1; %
				\kap = \radius * 0.2;
				\Angle1 = 30;
				\Angle2 = 90;
				\Angle0 = \Angle1 + 90;
				\Angle3 = \Angle2 - 90;
			};
			\begin{scope}
				\clip (0,0) -- (\Angle3:\radius) arc [start angle = \Angle3, end angle = \Angle0, radius = \radius] -- cycle;
				\fill[gray!40] (0,0) -- (\Angle0:\kap) -- +(\Angle1:\radius) -- (\Angle1:\radius) -- cycle;
				\fill[gray!40] (0,0) -- (\Angle3:\kap) -- +(\Angle2:\radius) -- (\Angle2:\radius) -- cycle;
			\end{scope}
			\begin{scope}
				\clip (0,0) -- (\Angle0:\kap) -- +(\Angle1:\radius) -- (\Angle1:\radius) -- cycle;
				\fill[gray!70] (0,0) -- (\Angle3:\kap) -- +(\Angle2:\radius) -- (\Angle2:\radius) -- cycle;
			\end{scope}
			\draw[very thick] (0,0) -- (\Angle1:\radius) (0,0) -- (\Angle2:\radius); 
			\node[anchor = north east] at (0,0) {$v(C)$};
			\draw[very thin] (\Angle1:\r1) arc [start angle=\Angle1, end angle=\Angle2, radius=\r1] node[pos=0.7, anchor = south west] {$\varphi(C)$};
			\curlybrace[very thin]{\Angle1}{\Angle2}{\r2} \node[anchor=south west] at (curlybracetipn) {$C$};
			\draw (0,0) -- (\Angle0:\radius);
			\tikzmath{\x=0.09*\radius;}
			\draw[very thin] (\Angle1:\x) -- + (\Angle0:\x) -- (\Angle0:\x);
			\curlybrace[very thin]{\Angle2}{\Angle0}{\r2} \node[anchor=south] at (curlybracetipn) {$\hat C^1$};
			\draw (0,0) -- (\Angle3:\radius);
			\tikzmath{\x=1.6*\x;}
			\draw[very thin] (\Angle3:\x) -- + (\Angle2:\x) -- (\Angle2:\x);
			\curlybrace[very thin]{\Angle3}{\Angle1}{\r2} \node[anchor=west] at (curlybracetipn) {$\hat C^2$};
			\tikzmath{\r3 = \r2+1;}
			\curlybrace[very thin]{\Angle3}{\Angle2}{\r3} \node[anchor = south west] at (curlybracetipn) {$E^2(C)$};
			\tikzmath{\r4 = \r3 + 0.2;}
			\curlybrace[very thin]{\Angle1}{\Angle0}{\r4} \node[anchor=south] at (curlybracetipn) {$E^1(C)$};
			\fill[black] (0,0) circle (0.1);
		\end{tikzpicture}
		\end{lrbox}
		\resizebox{0.5\textwidth}{!}{\usebox\figbox}
		&
		\begin{lrbox}{\figbox}
			\begin{tikzpicture} [thick,record path/.style={/utils/exec=\tikzset{bent brace/.cd,#1},
			    decorate,decoration={markings,mark=at position 0 with
			    {\setcounter{bracep}{1}%
			    \path(0pt,{\pgfkeysvalueof{/tikz/bent brace/scale}*\pgfkeysvalueof{/tikz/bent brace/raise}})
			     coordinate (bracep-\number\value{bracep});
			     \pgfmathsetmacro{\mystep}{(\pgfdecoratedpathlength-4pt)/int(1+(\pgfdecoratedpathlength-4pt)/2pt)}
			     \xdef\mystep{\mystep}},
			    mark=between positions 2pt and {\pgfdecoratedpathlength-2pt} step \mystep pt with {\stepcounter{bracep}
			    \pgfmathtruncatemacro{\itest}{
			    ifthenelse(abs(\number\value{bracep}*2pt/\pgfdecoratedpathlength-0.5)<1pt/\pgfdecoratedpathlength,1,0)}
			    \coordinate (bracep-\number\value{bracep}) at
			    (0,{\pgfkeysvalueof{/tikz/bent brace/scale}*(\pgfkeysvalueof{/tikz/bent brace/raise}+2pt+\itest*2pt)})
			    \ifnum\itest=1
			     node[transform shape,bent brace/nodes]{\pgfkeysvalueof{/tikz/bent brace/text}}
			    \fi;},
			    mark=at position 1 with {\stepcounter{bracep}
			    \coordinate (bracep-\number\value{bracep}) at (0,\pgfkeysvalueof{/tikz/bent brace/scale}*\pgfkeysvalueof{/tikz/bent brace/raise});}
			    }},brace/.style={insert path={plot[smooth,samples at={1,...,\number\value{bracep}},variable=\x]
			    (bracep-\x)}},
			    bent brace/.cd,raise/.initial=0pt,scale/.initial=1,text/.initial={},
			    nodes/.style={},node options/.code=\tikzset{bent brace/nodes/.append style=#1},
			    mirror/.code={\tikzset{bent brace/scale=-1}}]
			    \tikzmath{\radius = 2.5; 
				\r1 = \radius*0.3; %
				\r2 = \radius * 1.1; %
				\Angle1=130;
				\Angle2=-100;
				\Mid = (\Angle1+\Angle2)/2;
				\kap = \radius * 0.5;
			};
			\node[anchor = 15] at (0,0) {$v(C)$};
			\begin{scope}
				\clip (0,0) -- (\Angle1:\radius) arc [start angle=\Angle1, delta angle = -90, radius=\radius] -- cycle;
				\fill[gray!40] (0,0) -- ($(0,0)!1!-90:(\Angle1:\kap)$) -- + (\Angle1:\radius) -- (\Angle1:\radius) -- cycle;
			\end{scope}
			\begin{scope}
				\clip (0,0) -- (\Angle2:\radius) arc [start angle=\Angle2, delta angle = 90, radius=\radius] -- cycle;
				\fill[gray!40] (0,0) -- ($(0,0)!1!90:(\Angle2:\kap)$) -- + (\Angle2:\radius) -- (\Angle2:\radius) -- cycle;
			\end{scope}
			\fill[gray!70] (0,0) -- (\Angle1:\kap) arc [start angle=\Angle1, end angle=\Angle2,radius=\kap] -- cycle;
			\draw[very thin] (\Angle1:\r1) arc [start angle=\Angle1, end angle=\Angle2, radius=\r1] node[pos=0.5, anchor = -165,outer sep=4pt,inner sep=0pt] {$\varphi(C)$};
			\tikzmath{\boxsize = 0.3;}
			\tikzmath{\Anglebis = \Angle1 - 90;}
			\draw [very thick] (0,0) -- (\Angle1:\radius);
			\draw (0,0) -- (\Anglebis:\radius);
			\draw[very thin] (\Angle1:\boxsize) -- + (\Anglebis:\boxsize) -- (\Anglebis:\boxsize);
			\curlybrace[very thin]{\Anglebis}{\Angle1}{\r2} \node[anchor=south] at (curlybracetipn) {$E^1(C)$};
			\tikzmath{\Anglebis = \Angle2 + 90;}
			\draw[very thick] (0,0) -- (\Angle2:\radius);
			\draw (0,0) -- (\Anglebis:\radius);
			\draw[very thin] (\Angle2:\boxsize) -- + (\Anglebis:\boxsize) -- (\Anglebis:\boxsize);
			\curlybrace[very thin]{\Angle2}{\Anglebis}{\r2} \node[anchor=120] at (curlybracetipn) {$E^2(C)$};
			\tikzmath{\r = \radius *1.45;}
			\curlybrace[very thin]{\Angle2}{\Angle1}{\r} \node[anchor = west] at (curlybracetipn) {$C$};
			\fill[black] (0,0) circle (0.1);
		\end{tikzpicture}
		\end{lrbox}
		\resizebox{0.4\textwidth}{!}{\usebox\figbox}
	\end{tabular}
	\caption{Two examples of the notation for the definitions of the corner term; $C_\kappa$ is drawn in dark gray and a portion of $E_\kappa^1,E_\kappa^2$ in light gray, cf.\ Figure~\ref{fig:R_examples}. Note that in the example on the left $E_\kappa^1,E_\kappa^2$ extend outside of $C$ (and outside of $\Pi$); these extending portions are $C_\kappa^1$,$C_\kappa^2$.}
	\label{fig:corner}
\end{figure}

The point of all this is
\begin{equation}
	\begin{split}
		\int_{\Pi}
		P_{R(x)} (x,x;t)
		\dd x
		= &
		\int_{\Pi}
		P_0(t)
		\dd x
		+
		\sum_{E \in \cE}
		\int_{E_\kappa} \partial P_{H(E)} (x;t) \dd x
		\\ & \  +
		\sum_{C \in \cC}
		\left[ 
			\int_{C_\kappa} P^\angle_C (x;t) \dd x
			-
			\sum_{j=1,2}
			\int_{\hat C^j_\kappa} \partial P_{E^j} (x;t) \dd x
		\right]
	\end{split}
\end{equation}
(note that the integrals over $\hat C^j_\kappa$ cancel the part of the $E_\kappa$ integrals where $x \notin \Pi$) where
\begin{equation}
	P^\angle_C (x;t)
	:=
	\begin{dcases}
		P_C (x,x;t) - P_0(t)
		, & x \in C \setminus (\hat E^1 \cup \hat E^2)
		\\
		P_C (x,x;t) - P_{H(E_j)} (x,x;t)
		, & x \in (C \cap \hat E_j) \setminus \hat E_k, \ \left\{ j,k \right\} = \left\{ 1,2 \right\}
		\\
		\begin{aligned}
			& 	
			P_C (x,x;t) - P_0(t) 
			\\
			& \quad - \partial P_{ H(E_1)} (x;t) - \partial P_{H(E_2)},
		\end{aligned}
		& x \in C \cap \hat E_1 \cap \hat E_2,
	\end{dcases}
	\label{eq:Pangle_def}
\end{equation}
where $\hat E^j$ is the open quarter-plane with vertex $v(C)$ obtained by extending $E^j_\kappa$.
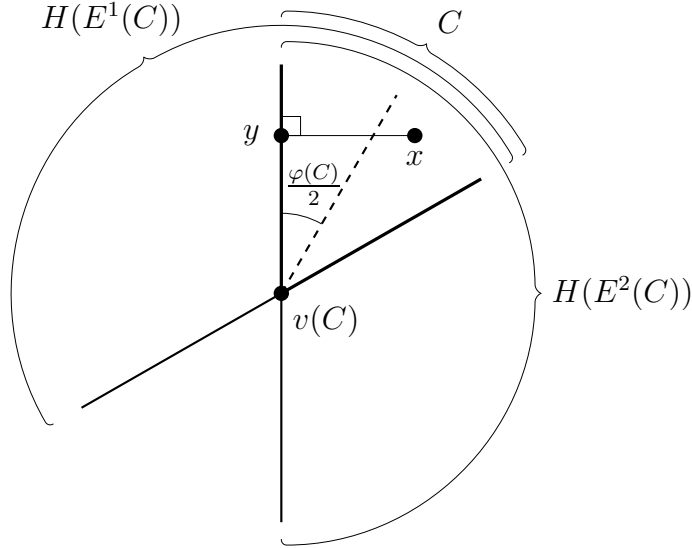
\begin{figure}[t]
	\centering
		\begin{lrbox}{\figbox}
			\begin{tikzpicture} [thick,record path/.style={/utils/exec=\tikzset{bent brace/.cd,#1},
			    decorate,decoration={markings,mark=at position 0 with
			    {\setcounter{bracep}{1}%
			    \path(0pt,{\pgfkeysvalueof{/tikz/bent brace/scale}*\pgfkeysvalueof{/tikz/bent brace/raise}})
			     coordinate (bracep-\number\value{bracep});
			     \pgfmathsetmacro{\mystep}{(\pgfdecoratedpathlength-4pt)/int(1+(\pgfdecoratedpathlength-4pt)/2pt)}
			     \xdef\mystep{\mystep}},
			    mark=between positions 2pt and {\pgfdecoratedpathlength-2pt} step \mystep pt with {\stepcounter{bracep}
			    \pgfmathtruncatemacro{\itest}{
			    ifthenelse(abs(\number\value{bracep}*2pt/\pgfdecoratedpathlength-0.5)<1pt/\pgfdecoratedpathlength,1,0)}
			    \coordinate (bracep-\number\value{bracep}) at
			    (0,{\pgfkeysvalueof{/tikz/bent brace/scale}*(\pgfkeysvalueof{/tikz/bent brace/raise}+2pt+\itest*2pt)})
			    \ifnum\itest=1
			     node[transform shape,bent brace/nodes]{\pgfkeysvalueof{/tikz/bent brace/text}}
			    \fi;},
			    mark=at position 1 with {\stepcounter{bracep}
			    \coordinate (bracep-\number\value{bracep}) at (0,\pgfkeysvalueof{/tikz/bent brace/scale}*\pgfkeysvalueof{/tikz/bent brace/raise});}
			    }},brace/.style={insert path={plot[smooth,samples at={1,...,\number\value{bracep}},variable=\x]
			    (bracep-\x)}},
			    bent brace/.cd,raise/.initial=0pt,scale/.initial=1,text/.initial={},
			    nodes/.style={},node options/.code=\tikzset{bent brace/nodes/.append style=#1},
			    mirror/.code={\tikzset{bent brace/scale=-1}}]
			\tikzmath{\radius = 3.0; 
				\r1 = \radius*0.35; %
				\kap = \radius * 0.2;
				\Angle1 = 30;
				\Angle2 = 90;
				\Angle0 = \Angle1 + 180;
				\Angle3 = \Angle2 - 180;
				\Mid = (\Angle1 + \Angle2)/2;
			};
			\draw[very thick] (0,0) -- (\Angle1:\radius) (0,0) -- (\Angle2:\radius); 
			\node[anchor = north west] at (0,0) {$v(C)$};
			\tikzmath{\r2 = (\radius*1.1)+0.4;}
			\curlybrace[very thin]{\Angle1}{\Angle2}{\r2} \node[anchor=south west] at (curlybracetipn) {$C$};
			\draw (0,0) -- (\Angle0:\radius);
			\draw (0,0) -- (\Angle3:\radius);
			\tikzmath{\r3 = \radius * 1.1;};
			\curlybrace[very thin]{\Angle3}{\Angle2}{\r3} \node[anchor = west] at (curlybracetipn) {$H(E^2(C))$};
			\tikzmath{\r4 = \r3 + 0.2;};
			\curlybrace[very thin]{\Angle1}{\Angle0}{\r4} \node[anchor=-50, outer sep = 4pt] at (curlybracetipn) {$H(E^1(C))$};
			\fill[black] (0,0) circle (0.1);
			\draw[dashed] (0,0) -- (\Mid:\radius);
			\draw[very thin] (\Angle2:\r1) arc [start angle = \Angle2, end angle = \Mid, radius=\r1] node[pos=0.5, anchor = -110] {$\frac{\varphi(C)}{2}$};
			\tikzmath{\r5=\radius*0.9;};
			\coordinate  (x) at (50:\r5);
			\fill[black] (x) circle (0.1);
			\node[anchor = north, outer sep = 2pt] at (x) {$x$};
			\coordinate (y) at ($(0,0)!(x)!(\Angle2:\radius)$);
			\fill[black] (y) circle (0.1);
			\draw [very thin] (x) -- (y);
			\draw [very thin] (y) rectangle + (0.25,0.25);
			\node[anchor = east, outer sep = 4pt] at (y) {$y$};
		\end{tikzpicture}
		\end{lrbox}
		\resizebox{0.7\textwidth}{!}{\usebox\figbox}
		\caption{Example of the estimates on $P^\angle_C(x;t)$ in the case $\varphi(C) \le \pi/2$; here $P^\angle_C = P_C + P_{\bR^2} - P_{H(E^1(C))} - P_{H(E^2(C))}$.  For the indicated $x$, $y$ is both the closest point in $C \triangle H(E^1(C))$ and the closest point in $H(E^2(C)) \triangle \bR^2$, and $|x-y| \ge \sin (\varphi(C)/2) \left|x - v(C)\right|$.  }
	\label{fig:corner_corrections}
\end{figure}
$P_C^\angle$ can be bounded using \cref{thm:PD_domain_change}, and it satisfies  $|P^\angle_C (x;t)| \le 2 \C t^{-1} e^{-k|x-v(C)|^2/t} $ 
for $k= k(\Pi) = \tfrac18 \min_{C \in \cC} \sin^2 \varphi(C)/2$:
for the first two cases in \cref{eq:Pangle_def} this is straightforward, 
for the last case (without loss of generality assuming $x$ is closer to $E_2$ than $E_1$, see \cref{fig:corner_corrections}) I write
\begin{equation}
	\begin{split}
		P^\angle_C (x;t)
		& =
		P_C (x,x;t) - P_0(t) 
		- \partial P_{ H(E_1)} (x;t) - \partial P_{H(E_2)}
		\\ &=
		\left( P_C (x,x;t) - P_{H(E_2)} (x,x;t) \right)
		+
		\left( P_{\bR^2} (x,x;t) - P_{H(E_1)} (x,x;t) \right)
		,
	\end{split}
	\label{eq:Pangle_diffs}
\end{equation}
and as can be seen in \cref{fig:corner_corrections} the distances involved in the las two differences are of the desired order.
From \cref{eq:6P_def} we see that the same is true of $\partial P(x;t)$ for $x \in \hat C^1 \cup \hat C^2$.
Consequently, we see that the integrals in
\begin{equation}
	a_2 (C)
	:=
	\int_{C} P^\angle_C (x;2 t) \dd x
	-
	\sum_{j=1,2}
	\int_{\hat C^j} \partial P_{E^j} (x;2 t) \dd x
	\label{eq:a2_C_def}
\end{equation}
are convergent; by rescaling $x$ we see that this is in fact independent of $t$, and by rotating and translating appropriately we see that in fact it depends only on $\varphi(C)$; an exact expression is given in \cite{MS67.eig}.
We also have
\begin{equation}
	a_2 (C)
	-
	\left[ 
		\int_{C_\kappa} P^\angle_C (x ;2t) \dd x
		-
		\sum_{j=1,2}
		\int_{\hat C^j_\kappa} \partial P_{E^j} (x;2 t) \dd x
	\right]
	=
	\cO\left( \kappa^{-1} e^{- k \kappa^2/t} \right);
	\label{eq:a2_C_rem}
\end{equation}
combining this with \cref{eq:ref_to_a1,eq:ref_to_a0} and letting $a_2(\Pi) = \sum_{C \in \cC} a_2(C)$ we get
\begin{equation}
	\int_{\Pi} P_{R_\Pi (x)} (x,x;2 t) \dd x
	=
	a_0(\Pi) t^{-1}
	+
	a_1(\Pi) t^{-1/2}
	+
	a_2(\Pi)
	+
	\cO\left( e^{- \delta^2 / t} \right)
	.
	\label{eq:heat_trace_reference}
\end{equation}
Since
\begin{equation}
	\int_{\Pi} P_{\Pi} (x,x;t) \dd x
	-
	\int_{\Pi} P_{R_\Pi (x)} (x,x;t) \dd x
	=
	\cO\left( e^{- \kappa^2  / t} \right)
	\label{eq:continuum_R_error}
\end{equation}
we obtain \cref{eq:ht_asympt}.

\section{Long time: convergence of the discrete and continuous heat traces}
\label{sec:convergence_to_continuous}

Let $\cB_s$ be a two dimensional Brownian bridge, shifted and rescaled so that $\cB_0 = \cB_t = x$; we then have
\begin{equation}
	P_{\Omega,\Sigma}(x,x;t)
	=
	P_0(t)
	\bE^t_x \left[ D_\Omega(\cB) e^{i \pi W_\Sigma(\cB;t)} \right]
	;
	\label{eq:bridge_PD}
\end{equation}
where $D_\Omega(\cB)$ is the indicator function of the event that $\cB_s \in \Omega$ for all $s \in [0,t]$;
recall $P_0(t) = P_{\bR^2} (x,x;t) = 1/2 \pi t$.
Similarly, letting $\wt \cB_{m}$ be the process obtained by conditioning the random walk $\wt \cW_{m}$ to return to its starting point after $2n$ steps,
\begin{equation}
	\wt P_{\Omega,\Sigma}(x,x;2n)
	=
	\wt P_0(2n)
	\bE^{n}_x \left[ D_\Omega(\wt \cB) e^{i \pi W_\Sigma(\wt \cB;2n)} \right]
	,
	\label{eq:bridge_tPD}
\end{equation}
where $\wt P_0 (n) = \wt P_{\bR^2}(x,x;n)$ is the probability that a simple random walk on $\bZ^2$ returns to its starting point after $n$ steps; representing this in terms of two independent $1D$ walks as in \cref{eq:independent_RW_trick},
\begin{equation}
	\wt P_0 (2n) 
	=
	\left[ 2^{-2n} \binom{2n}{n}  \right]
	=
	\frac{1}{\pi n} + \cO\left( \frac{1}{n^2} \right)
	=
	2 P_0 (n) + \cO\left( \frac{1}{n^2} \right)
	\label{eq:P0_comparison}
\end{equation}
for large $n$ using Stirling's formula.

Combining \cref{eq:bridge_PD,eq:bridge_tPD} gives
\begin{equation}
	\begin{split}
		P_{\Omega,\Sigma}(x,x;n)
		-
		\frac12
		\wt 
		P_{\Omega,\Sigma}(y,y;2 n)
		= &
		P_0(n)
		\left\{ 
			\bE^n_x \left[ D_\Omega(\cB) e^{i \pi W_\Sigma(\cB;n)} \right]
			-
			\bE^n_y \left[ D_\Omega(\wt \cB) e^{i \pi W_\Sigma(\wt \cB;2n)} \right]
		\right\}
		\\ & +
		\left[ 
			1 -
			\frac{ 2 P_0(n)}{\wt P_0(2n)}
		\right]
		\wt 
		P_{\Omega,\Sigma}(y,y;2 n)
		;
	\end{split}
	\label{eq:layer_diff_1}
\end{equation}
note that the endpoints $x$ and $y$ may be different, and the different time scales for the two processes, which are consistent with the asymptotic behavior of their variances for long times,
and the factor of $1/2$ in front of the discrete probabilities which compensates for the reducibility of the discrete random walk.

This is helpful because the two bridge processes can be coupled using a variant of the dyadic coupling, using a result of \cite{LF07}:

\begin{theorem}
	There is a constant $\Cl{dyadic} > 0$ such that
	for any $y \in \bZ^2$, $x \in \bR^2$ with $|x - y|_\infty \le 1/2$ and any integer $n > 1$,
	there exists a coupling of $\cB_s$ and $\wt \cB_s$ such that
	\begin{equation}
		\bP_{x,y}^{n} \left[ \sup_{s \in [0,n]} \left| \cB_{s} - \wt\cB_{2s}\right| \ge \Cr{dyadic} \log  n \right]
		\le
		\Cr{dyadic} n^{-30}
		.
		\label{eq:bridge_Skor}
	\end{equation}
	\label{thm:bridge_Skor}
\end{theorem}
\begin{proof}
	The corresponding statement with $x=y=0$ is obtained by rescaling \cite[Corollary~3.2]{LF07}, which is in turn based on decomposing $\wt \cB$ into two independent diagonal random walks as in the proof of \cref{thm:PDt_domain_change} to use the one dimensional construction of~\cite{KMT75}; 
	translating these processes to the desired starting points then gives the desired result up to a change in the constant, since $|x-y| \le 1/2 \le \log n$.
\end{proof}

Using this representation,
\begin{multline}
	\bE^{n}_x \left[ D_\Omega(\cB) e^{i \pi W_\Sigma(\cB;n)} \right]
	-
	\bE^{n}_y \left[ D_\Omega(\wt \cB) e^{i \pi W_\Sigma(\wt \cB;2n)} \right]
	\\ 
	=
	\bE_{x,y}^{n}
	\left[ 
		D_\Omega(\cB) e^{i \pi W_\Sigma(\cB;n)} 
		-
		D_\Omega(\wt \cB) e^{i \pi W_\Sigma(\wt \cB;2n)} 
	\right]
	.
\end{multline}
For any $\delta > 0$, restricting to the situation where the two processes remain within a distance $\delta$ of each other for the whole time interval $[0,n]$, the difference on the right hand side can be nonzero only if both $\cB$ and $\wt \cB$ come within a distance $\delta$ of $\Omega^\comp$ but at least one of them remains entirely in $\Omega$.  
This in turn requires that in the same time interval $\cB_s$ leaves
\begin{equation}
	\Omega^-(\delta)
	:=
	\left\{ 
		x \in \Omega
		: \ 
		d(x,\Omega^\comp) > \delta
	\right\}
	,
	\label{eq:Omega_minus_def}
\end{equation}
but does not leave
\begin{equation}
	\Omega^+(\delta)
	:=
	\left\{ 
		x \in \bR^2
		: \ 
		d(x,\Omega) < \delta
	\right\}
	.
	\label{eq:Omega_plus_def}
\end{equation}
Using this observation with $\delta = \Cr{dyadic} \log n$ and noting that $|1 - 2 P_0(n)/\wt P_0(2n)| \le \Cl{plane_disc}/n$ (recall \cref{eq:P0_comparison}), \cref{eq:layer_diff_1} gives 
\begin{equation}
	\begin{split}
		&
		\left|
			P_\Omega (x,x;n)
			-
			\frac12
			\wt P_\Omega (y,y;2n)
		\right|
		\\ & \quad
		\le
		\lim_{r \to 0^+}
		\frac{1}{\pi r^2}
		\bP_x
		\left[ 
			\cW_t \in B_r(x), \ 
			T_{\Omega^-(\delta)} < n < T_{\Omega^+(\delta)}
		\right]
		+
		\C/n^2
		\\ & \quad
		=
		\left\{ 
			P_{\Omega^+(\delta)}(x,x;2n)
			-
			P_{\Omega^-(\delta)}(x,x;2n)
		\right\}
		+
		\Clast/n^2
		.
	\end{split}
	\label{eq:layer_diff}
\end{equation}
The first term on the right hand side appears to be difficult to estimate precisely when $\Omega$ is not convex and when the distance of $x$ from the boundary is small compared to $\sqrt t$; fortunately some distinctly suboptimal bounds (\cref{lem:PD_Beurling,lem:layer_hard} below) will suffice to give a serviceable estimate.  
The starting point are some estimates on hitting times, some of which are admittedly crude versions of well known results, of which I provide basic proofs for completeness.
\begin{lemma}
	For all $\tau > 0$, $\Omega \subset \bR^2$, and all $x\in \bR^2$,
	\begin{equation}
		\bP_x\left[ T_{\Omega} \ge \tau  \right]
		\le
		e^{1 - \C \tau /|\Omega|}
		\label{eq:survival_simple}
	\end{equation}
	with $\Clast = 2 \pi / e$.
	\label{lem:survival_simple}
\end{lemma}
\begin{proof}
	Using the Markov property of Brownian motion,
	\begin{equation}
		\begin{split}
			\bP_x \left[ T_\Omega \ge \tau \right]
			& \le
			\bP_x \left[ \cW_{\tau/n} \in \Omega, \ \cW_{2\tau/n} \in \Omega, \ \dots \right]
			\\ &\le 
			\left\{ \sup_{y \in \Omega} \bP_y \left[ \cW_{\tau/n} \in \Omega \right] \right\}^n
			\le
			\left\{ \frac{n}{2 \pi \tau} |\Omega| \right\}^n
		\end{split}
	\end{equation}
	for any positive integer $n$; as long as $\tau \ge e |\Omega|/2\pi$ we can always choose
	\begin{equation}
		\frac{e^{-1} 2 \pi \tau}{|\Omega|}
		-1
		\le 
		n
		\le 
		\frac{e^{-1} 2 \pi \tau}{|\Omega|}
	\end{equation}
	to obtain \cref{eq:survival_simple}, and for smaller $\tau$ the bound is trivially correct.
\end{proof}

This can be used to obtain another estimate which is more useful close to the boundary.  I will state these estimates for domains of the following type:
\begin{definition}
	An \emph{$S(\rho)$ set}
	is an open set $\Omega \subset \bR^2$ such that any point $y \in \bR^2 \setminus \Omega$ is an endpoint of a line segment $\ell_y$ of length $\rho$ which does not intersect $\Omega$.
	\label{def:Srho}
\end{definition}
Note that the polygon $\Pi$ we consider is an $S(1)$ set, and $L\Pi$ is an $S(L)$ set;
$\Pi^+(\delta)$ and $\Pi^-(\delta)$ (respectively $(L\Pi)^+(\delta)$ and $(L\Pi)^-(\delta)$ are $S(1/2)$ (resp.\ $S(L/2)$) for $\delta < \Cl{delta_layer_1}$ (resp.\ $\delta/L< \Clast$) for some $\Pi$-dependent $\Clast$.

\begin{lemma}
	There exists $\Cl{survival} > 0$ such that 
	if $\tau >0$, $\Omega$ is an $S(\sqrt\tau)$ set, and $x \in \Omega$, then
	\begin{equation}
		\bP_x\left[ T_{\Omega} > \tau \right]
		\le 
		\Clast
		\frac{R^{1/2} (\log \tau - 2 \log R)^{1/4}}{\tau^{1/4}}
		\label{eq:survival_bound}
	\end{equation}
	with $R$ the distance between $x$ and $\partial\Omega$, as long as $R \le e^{-2/e}\sqrt \tau $.
	\label{lem:survival_bound}
\end{lemma}

\begin{proof}
	Let $y$ be a point in $\partial\Omega$ with $|x-y|=R$, and let $\ell_y$ be the line segment of length $\sqrt\tau$ associated with $y$ by \cref{def:Srho}.  Since $\ell \subset \Omega^\comp$, trivially
	\begin{equation}
		\bP_x\left[ T_{\Omega} > \tau \right]
		\le
		\bP_x\left[ T_{B_\rho(y)} > \tau \right]
		+
		\bP_x\left[ T_{\ell^\comp} > T_{B_\rho(y)} \right]
	\end{equation}
	for any $\rho > 0$.  For $\rho \le \sqrt \tau$ the second term on the right hand side is the solution of a harmonic problem which can be solved explicitly:  for $z$ a unit complex number depending on $\ell$ and $y$,
	\begin{equation}
		\bP_x\left[ T_{\ell^\comp} > T_{B_\rho(y)} \right]
		=
		\frac{4}{\pi}
		\Re  \arctan \left( z \sqrt{\frac{x-y}{\rho}} \right)
		\le
		\frac{4}{\pi}
		\sqrt{\frac{R}{\rho}},
		\label{eq:pseudo_Beurling}
	\end{equation}
	and using this along with \cref{lem:survival_simple} we obtain
	\begin{equation}
		\bP_x\left[ T_{\Omega} > \tau \right]
		\le
		\exp\left( 1 - \frac{\tau}{e \rho^2} \right)
		+
		\frac{4}{\pi}
		\sqrt{\frac{R}{\rho}}
		.
	\end{equation}
	Choosing $\rho = 2 e^{-1/2} \sqrt{\tau/(\log \tau - 2 \log R)}$ (which satisfies $\rho \le \sqrt \tau$ under the assumption $R \le e^{-2/e} \sqrt\tau$), this gives \cref{eq:survival_bound}.
\end{proof}
This can then be use to obtain an estimate on the heat kernel which will be useful for estimating $P_{(L\Pi)^+(\delta)}$, and hence the difference, near the boundary where any Brownian motion is very likely to exit $(L\Pi)^+(\delta)$:
\begin{lemma}
	There exists $\Cl{PD_Beurling}>0$ such that
	if $t >0$, $\Omega$ is an $S(\sqrt{\tau})$ set for some $\tau \le t/3$, and $x \in \Omega$, then
	\begin{equation}
		P_\Omega(x,x;t)
		\le
		\Clast
		\frac{R \left( \log \tau - 2 \log R 
		\right)^{1/2}}{ \tau^{1/2} t}
		\label{eq:PD_Beurling_bound}
	\end{equation}
	with $R$ the distance between $x$ and $\partial\Omega$, as long as $R \le e^{-2/e}\sqrt{\tau} $.
	\label{lem:PD_Beurling}
\end{lemma}
\begin{proof}
	Noting that
	\begin{equation}
		\begin{split}
			P_\Omega (x,x;t)
			&=
			\int_{\Omega} \dd y
			\int_{\Omega} \dd z
			P_\Omega (x,y;\tau)
			P_\Omega (y,z;\tau)
			P_\Omega (z,x;t-2\tau)
			\\ & \le
			\left[ \int_\Omega \dd y P_\Omega (x,y;\tau) \right]^2
			\sup_{z,w \in \Omega} P_{\bR^2} (z,w;t-2\tau)
			\\ & \le
			\left\{ \bP_x \left[ T_{\Omega^\comp} > t/3 \right] \right\}^2
			\frac{3}{2 \pi t}
		\end{split}
	\end{equation}
	using the assumption $\tau \le t/3$ in the last inequality,
	the result follows from Lemma~\ref{lem:survival_bound}.
\end{proof}

When the starting point is far from the boundary, I can instead obtain a useful estimate via the probability of a Brownian motion, having left $(L\Pi)^-(\delta)$, manages to travel far enough to get back to its starting point without leaving $(L\Pi)^+(\delta)$, whose boundary is now at a distance of order $\delta$.

\begin{lemma}
	There exists $\Cl{layer_hard}$ for which the following holds.
	For any $S(\rho)$ set $\Xi$, any open measurable $\Omega \subset \Xi$ with $\max_{y \in \partial\Omega} \min_{z \in \partial\Xi} |y-z| = \delta'$, and any $x \in \Omega$ with $\dist(x,\partial\Omega) = R$,
	\begin{equation}
			P_\Xi(x,x;t) - P_\Omega(x,x;t)
			\le
			\frac{\Clast}{R^2}
			\sqrt{\frac{\delta'}{\min(R,2\rho)}}
			.
		\label{eq:layer_hard}
	\end{equation}
	\label{lem:layer_hard}
\end{lemma}
\begin{proof}
	Firstly, using the strong Markov property and conditioning on the joint distribution of $T_\Omega$ and $\cW_{T_\Omega}$,
	\begin{equation}
		\begin{split}
			P_\Xi(x,x;t) - P_\Omega(x,x;t)
			& =
			\lim_{r \to 0^+} \frac{1}{\pi r^2}
			\bP_x\left[ \left|\cW_t -x\right| < r, \ T_\Omega < t < T_\Xi \right]
			\\ & =
			\bE_x\left[ \ind_{[0,t)}(T_\Omega) P_\Xi (\cW_{T_\Omega},x;t-T_\Omega) \right]
			\\ & \le
			\sup_{y \in \partial\Omega}
			\sup_{s > 0}
			P_\Xi (y,x;s)
			;
		\end{split}
		\label{eq:boundary_diff_bound}
	\end{equation}
	similarly, conditioning on the event that $\cW_s$ leaves a certain disk $B_q(z)$ which contains $y$ but not $x$ before exiting $\Xi$,
	\begin{equation}
		\begin{split}
			P_\Xi (y,x;s)
			& =
			\bE_y\left[ \ind \left\{ T_{B_q(z)} \le \min(s,T_\Xi) \right\} P_\Xi(\cW_{T_{B_q(z)}},x;s-T_{B_q(z)}) \right]
			\\ & \le
			\bP_y\left[ T_{B_q(z)} \le T_\Xi \right] \sup_{w \in B_q(z)} \sup_{u > 0} P(w,x;u).
		\end{split}
		\label{eq:boundary_PDxy_bound}
	\end{equation}
	Letting $z$ be a point on the boundary of $\Xi$ with $|z - y| \le \delta'$, for all $q \le \rho$ we can use the assumption that $\Xi$ is an $S(\rho)$ set as in \cref{eq:pseudo_Beurling} to bound $\bP_y\left[ T_{B_q(z)} \le T_\Xi \right] \le (4/\pi) \sqrt{q/\delta'}$, and noting $P(w,x;u) \le 1/(e \pi |w-x|^2) \le 1/(e \pi (R-q)^2)$, \cref{eq:layer_hard} follows by choosing $q = \min(\rho,R/2)$.
\end{proof}

\begin{lemma}
	There exist $\Cl{layer_delta}, \ \Cl{layer}, \ \Cl{layer_time}, \ \Cl{layer_max_time} > 0$, depending on $\Pi$, such that for $\delta \le \Cr{layer_delta}L$, and $\dist(x,\partial\Pi^+(\delta)) = R$,
	\begin{equation}
		\left[ P_{(L\Pi)^+(\delta)} - P_{(L\Pi)^-(\delta)} \right](x,x;t)
		\le
		\Cr{layer}
		\begin{cases}
			\frac{(\log t - 2 \log R)^{1/2}R}{t^{3/2}},
			&
			R  < t^{3/7} 
			\\
			\frac{\sqrt{\delta}}{R^{5/2}},
			&
			t^{3/7} \le R < L
			\\
			\frac{\sqrt{\delta}}{R^{2}L^{1/2}},
			&
			R \ge L
		\end{cases}
		\label{eq:layer_main_early}
	\end{equation}
	whenever $ \Cr{layer_time} \le t \le \tfrac34 L^2$, and
	\begin{equation}
		\left[ P_{(L\Pi)^+(\delta)} - P_{(L\Pi)^-(\delta)} \right](x,x;t)
		\le
		\Cr{layer}
		\begin{cases}
			\frac{(\log L - 2 \log R)^{1/2}R}{L t},
			&
			R  < t^{2/7} L^{2/7} 
			\\
			\frac{\sqrt{\delta}}{R^{5/2}},
			&
			t^{2/7} L^{2/7} \le R < L
			\\
			\frac{\sqrt{\delta}}{R^{2}L^{1/2}},
			&
			R \ge L
		\end{cases}
		\label{eq:layer_main_late}
	\end{equation}
	when $\tfrac34 L^2 \le t \le \Cr{layer_max_time} L^{5/2}$.
	\label{lem:layer_main}
\end{lemma}

\begin{proof}
	First of all, there exist $\Cr{layer_delta}, \Cl{layer_delta_prime}<\infty$ such that,
	for all $\delta \le \Cr{layer_delta}$,
	$\Pi^+(\delta)$ is an $S(1/2)$ set and 
	$\max_{y \in \partial\Pi^-(\delta)} \min_{z \in \partial\Pi^+(\delta)} |y-z| \le \Cr{layer_delta_prime} \delta$;
	rescaling, this implies that $(L\Pi)^+(\delta)$ is $S(L/2)$ 
	and $\max_{y \in \partial(L\Pi)^-(\delta)} \min_{z \in \partial(L\Pi)^+(\delta)} |y-z| \le \Cr{layer_delta_prime} \delta$
	for $\delta \le \Cr{layer_delta} L$.

	Then by choosing $\Cr{layer_time}$ large enough \cref{lem:PD_Beurling} applies in the first case on the right hand side of \cref{eq:layer_main_early} with $\tau = t/3$, giving the advertised bound,
	and for $\Cr{layer_max_time}$ small enough likewise for the first case in \cref{eq:layer_main_late} with $\tau = L^2/4$.
	The other cases follow from \cref{lem:layer_hard} with $\rho = L/2$.
\end{proof}
Integrating in $x$, we obtain
\begin{corollary}
	There exists $\C$ such that 
	\begin{equation}
		\int_{(L\Pi)^+(\delta)} 
		\left[ P_{(L\Pi)^+(\delta)} - P_{(L\Pi)^-(\delta)} \right](x,x;t)
		\dd x
		\le
		\Clast
		\left( \sqrt{\delta} + \log L \right)
		\begin{cases}
			L t^{-9/14}
			, &
			t \le \tfrac34 L^2
			\\
			L^{4/7} t^{-3/7}
			, &
			t > \tfrac34 L^2
		\end{cases}
	\end{equation}
	whenever $\delta \le \Cr{layer_delta}L$ and $ \Cr{layer_time} \le t \le \Cr{layer_max_time} L^{5/2}$.
	\label{lem:layer_Pi_integral}
\end{corollary}

We now arrive at the conclusion of this section:
\begin{theorem}
	For each $\Pi$ there exists $\Cl{disc_int_err} < \infty$
	such that,
	whenever $K \in \bN$ and ${e^{9/2}\le K \le L}$,
	\begin{equation}
		\left|
			\int_{K^2}^\infty
			\int_{L\Pi} P_{L \Pi} (x,x;t) \dd x
			\frac{\dd t}{ t}
			-
			\sum_{n=  K^2}^\infty
			\sum_{y \in L \Pi \cap \bZ^2} \frac{\wt P_{L \Pi} (y,y;2n)}{2n}
		\right|
		\le 
		\Cr{disc_int_err}
		\left( \frac{L \log L}{K^{9/7}} + \frac{ L^2}{K^4} \right)
		.
		\label{eq:discretization_err_integrated}
	\end{equation}
	\label{thm:discretization_err_integrated}
\end{theorem}
\begin{proof}
	Using \cref{eq:layer_diff}, we have
	\begin{equation}
		\begin{split}
			&
			\left|
			\int_{L\Pi} P_{L \Pi} (x,x;n) \dd x
			-
			\frac12
			\sum_{y \in L \Pi \cap \bZ^2} \wt P_{L \Pi} (y,y;2n)
			\right|
			\\ & \qquad \le
			\int_{L \Pi}
			\left\{ 
				P_{(L\Pi)^+(\delta)}(x,x;n)
				-
				P_{(L\Pi)^-(\delta)}(x,x;n)
			\right\}
			\dd x
			+
			\Cr{plane_disc} \frac{|\Pi|L^2}{t^2}
			\\ & \qquad \le
			\int_{(L \Pi)^-(\delta)}
			\left\{ 
				P_{(L\Pi)^+(\delta)}(x,x;n)
				-
				P_{(L\Pi)^-(\delta)}(x,x;n)
			\right\}
			\dd x
			\\ & \qquad \qquad
			+ 
			\int_{L \Pi \setminus (L \Pi)^-(\delta)}
			\left\{ 
				P_{(L\Pi)^+(\delta)}(x,x;n)
			\right\}
			\dd x
			+
			\C \frac{L^2}{n^2}
		\end{split}
	\end{equation}
	with $\delta = \Cr{dyadic} \log n$.
	For $n \le \Cr{layer_max_time} L^{5/2}$, we can bound the first integral using \cref{lem:layer_Pi_integral}, and the second using $P_{(L\Pi)^+(\delta)}(x,x;n) \le 1/4\pi n$, giving 
	\begin{equation}
			\left|
			\int_{L\Pi} P_{L \Pi} (x,x;n) \dd x
			-
			\frac12
			\sum_{y \in L \Pi \cap \bZ^2} \wt P_{L \Pi} (y,y;2n)
			\right|
			\le
			\C
			\left[ 
				\frac{L^2}{n^2}
				+
				\log L 
				\begin{cases}
					L n^{-9/14}
					, &
					n \le \tfrac34 L^2
					\\
					L^{4/7} n^{-3/7}
					, &
					n > \tfrac34 L^2
				\end{cases}
			\right]
			.
	\end{equation}
	For $n > \Cr{layer_max_time} L^{5/2}$ a stronger bound follows easily from \cref{lem:survival_bound}; the error made by replacing the $t$ integral in \cref{eq:discretization_err_integrated} with a sum can easily be seen to be $\cO(L^2/K^4)$ using the Euler-Maclaurin formula, and so \cref{eq:discretization_err_integrated} follows by summing the above estimates.
\end{proof}

\section{The reference geometry terms in the discrete case}
\label{sec:discrete_ref}

In this section I control the difference between the discrete heat kernel on $\Pi$ and in the reference geometry, and relate the latter to the terms $\alpha_0, \ \alpha_1,\  \alpha_2$ in \cref{eq:main}.
To take into account the presence of a lattice scale, we define
\begin{equation}
	\wt P^\R_{L \Pi} (x;n)
	:=
	\wt P_{L R_\Pi (L^{-1} x)}(x,x;n)
	.
	\label{eq:tPR_def}
\end{equation}

We now introduce counterparts of $a_0(\Pi), \ a_1(\Pi), a_2(\Pi)$; however their time dependence is not so simple as in the continuum case.  
For $H$ a half-plane we let
\begin{equation}
	\partial \wt P_H (x;n)
	:=
	\wt P_H (x,x;n)
	-
	\wt P_0 (n)
	\label{eq:6P_tilde_def}
\end{equation}
(cf.\ \cref{eq:6P_def})
and define $\wt P^\angle_C$ by the analogue of \cref{eq:Pangle_def}.  We then have 
\begin{equation}
	\begin{split}
		\sum_{x \in L\Pi \cap \bZ^2}
		\wt P^\R_{L\Pi} (x;n)
		= &
		\left| \Omega \cap \bZ^2 \right| \wt P_0(n)
		+
		\sum_{E \in \cE }
		\sum_{x \in LE_\kappa \cap \bZ^2}
		\partial \wt P_{H(E)} (x;n)
		\\ & +
		\sum_{C \in \cC}
		\left[ 
			\sum_{x \in L C_\kappa \cap \bZ^2}
			\wt P^\angle_C (x;n)
			-
			\sum_{j=1,2}
			\sum_{x \in L \hat C^j_\kappa \cap \bZ^2} 
			\partial \wt P_{H(E^j)} (x;n)
		\right]
		.
	\end{split}
	\label{eq:PR_tilde_decomp}
\end{equation}
To extract the asymptotics we are interested in, let
\begin{equation}
	A_0 (\Omega,n)
	:=
	\# \left( \Omega \cap \bZ^2 \right) 
	\wt P_0(n)
	,
	\label{eq:A0_def}
\end{equation}
\begin{equation}
	A_1(L\Pi,n)
	:=
	\sum_{E \in \cE }
	\hat A_1(E,L,n)
	:=
	\sum_{E \in \cE }
	\sum_{x \in LE \cap \bZ^2}
	\partial \wt P_{L H(E)} (x;n)
	,
	\label{eq:A1_def}
\end{equation}
\begin{equation}
	A_2(L\Pi,n)
	:=
	\sum_{C \in \cC}
	\hat A_2(C,n)
	:=
	\sum_{C \in \cC}
	\left[ 
		\sum_{x \in L C \cap \bZ^2}
		\wt P^\angle_C (x;n)
		-
		\sum_{j=1,2}
		\sum_{x \in L \hat C^j \cap \bZ^2} 
		\partial \wt P_{H(E^j)} (x;n)
	\right]
	.
	\label{eq:A2_def}
\end{equation}
Note that $A_2$ is actually independent of $L$, $A_2 (L\Pi,n) = A_2(\Pi,n)$, since rescaling each wedge $C \in \cC$ by an integer factor gives the same result; 
it is also invariant under the symmetries of $\bZ^2$, and in fact $\hat A_2(C,n)$ is evidently determined by the opening angle of $C$ and its orientation relative to the lattice.
Analogously $A_1(L\Pi,n) = L A_1(\Pi,n)$, and $\hat A_1(E,L,n)$ is fixed by the length and orientation of the edge of $L\Pi$ corresponding to $E$.

Subtracting these terms from \cref{eq:PR_tilde_decomp}, we can repeat the analysis of \cref{sec:ref_geom} to write the remainder in terms of differences between (discrete) heat kernels on domains which differ at distance at least $\kappa L$ from the points of interest, which can be bounded using \cref{thm:PDt_domain_change}, giving the counterpart of \cref{eq:heat_trace_reference}
\begin{equation}
	\begin{split}
		&
		\left|
		\sum_{x \in L\Pi \cap \bZ^2}
		\wt P^\R_{L \Pi} (x;n)
		-
		A_0(L\Pi,n)
		-
		A_1(L\Pi,n)
		-
		A_2(L\Pi,n)
		\right|
		\\ & \qquad 
		\le
		\C 
		e^{-\C L^2/n}
	\end{split}
\end{equation}
for $n \le L^2$.
Applying \cref{thm:PDt_domain_change} again to bound the difference between $\wt P^\R_{L\Pi}$ and $\wt P_{L\Pi}$ gives an estimate of similar magnitude,
\begin{equation}
	\begin{split}
		&
		\left|
		\sum_{x \in L\Pi \cap \bZ^2}
		\wt P_{L \Pi} (x;n)
		-
		A_0(L\Pi,n)
		-
		A_1(L\Pi,n)
		-
		A_2(L\Pi,n)
		\right|
		\\ & \qquad 
		\le
		\C 
		e^{-\C L^2/n}
		,
	\end{split}
	\label{eq:PDt_to_A}
\end{equation}
and consequently
\begin{equation}
	\begin{split}
		&
		\left|
		\sum_{n=1}^{K^2-1}
		\frac1{2n}
		\left\{ 
			\sum_{x \in L\Pi \cap \bZ^2}
			\wt P_{L \Pi} (x;2 n)
			-
			A_0(L\Pi,2n)
			-
			A_1(L\Pi,2n)
			-
			A_2(L\Pi,2n)
		\right\}
		\right|
		\\ & \qquad
		\le
		\C %
		\exp \left( - \C \frac{L^2}{K^2} \right)
	\end{split}
	\label{eq:discrete_reference_error}
\end{equation}
for all $1 < K \le L$.

Let us now compare these contributions to their continuum counterparts.  For $A_0$ this is quite simple: recalling \cref{eq:ref_to_a0}, 
\begin{equation}
	\begin{split}
		&
		L^2 \frac{a_0(\Pi)}{n/2}
		-
		\frac12
		A_0 (L\Pi,2n)
		=
		L^2 |\Pi|  P_0 (n)
		-
		\frac12 
		\# \left( L\Pi \cap \bZ^2 \right)
		\wt P_0 (2n)
		\\ & \hphantom{LLLL}
		=
		L^2 |\Pi| \left[ P_0 (n)  - \tfrac12 \wt P_0 (2n) \right]
		+ \frac12 \left[L^2 |\Pi| - \# (L \Pi \cap \bZ^2) \right] \left[\wt P_0(2n)  \right]
		\\ & \hphantom{LLLL}
		=
		\cO\left( \frac{L^2}{n^2} + \frac{L}{n} \right)
		, \quad \frac{L}{n} \to \infty
	\end{split}
	\label{eq:A0_diff}
\end{equation}
using Pick's formula to estimate the difference in areas.

For $A_1$, recalling \cref{eq:ref_to_a1} we have
\begin{equation}
	L \frac{a_1(\Pi)}{(n/2)^{1/2}}
	- 
	\frac12
	A_1(L\Pi,2n) 
	=
	L
	\sum_{E \in \cE}
	\left[ 
		\int_{E}
		\partial P_{H(E)} (x;n)
		\dd x
		-
		\frac12
		\sum_{y \in E \cap \bZ^2}
		\partial \wt P_{H(E)} (y;2n)
	\right]
	;
	\label{eq:A1_diff}
\end{equation}
This difference can be bounded using a similar approach to \cref{sec:convergence_to_continuous}, with a few variations and added elements.  
To begin with, note that
\begin{equation}
	\begin{split}
		&
		\partial P_{H(E)} (x;n)
		-
		\frac12
		\partial \wt P_{H(E)} (y;2n)
		=
		P_{H(E)}(x,x;n)
		-
		P_0 (n)
		-
		\frac12
		\wt P_{H(E)}(y,y;2n)
		+
		\frac12
		\wt P_0(2n)
		\\ & \quad
		=
		P_0(n)
		\left\{ 
			\bP_y^n \left[ \exists s \in (0,2n): \ \wt \cB_s \notin H(E) \right]
			-
			\bP_x^n \left[ \exists s \in (0,n): \ \cB_s \notin H(E) \right]
		\right\}
		\\ & \qquad \qquad
		+
		\left[ 
			1 -
			\frac{ 2 P_0(n)}{\wt P_0(2n)}
		\right]
		\partial \wt 
		P_{H(E)}(y,y;2n)
	\end{split}
	\label{eq:HP_layer_1}
\end{equation}
(cf.\ \cref{eq:layer_diff_1}).  The difference in probabilities on the right hand side of \cref{eq:HP_layer_1} is bounded by the probability that exactly one of the bridge processes leaves the half plane $H(E)$; as in \cref{eq:layer_diff}, we can bound this difference using \cref{thm:bridge_Skor}, giving
\begin{equation}
		\left|
		\partial P_{H(E)} (x;n)
		-
		\frac12
		\partial \wt P_{H(E)} (y;2n)
		\right|
		\le
		P_{[H(E)]^+(\delta)}(x,x;n)
		-
		P_{[H(E)]^-(\delta)}(x,x;n)
		+
		\C/n^2
	\label{eq:HP_layer_2}
\end{equation}
for all $n > 1$, $|x-y|_\infty \le 1/2$, with $\delta = \Cr{dyadic} \log n$.
This difference can be bounded as in \cref{lem:layer_main} using \cref{lem:survival_bound,lem:PD_Beurling,lem:layer_hard}, with the simplification that the complement of $[H(E)]^+(\delta)$ always contains an infinite half-line touching any point on its boundary (in terms of \cref{def:Srho}, it is $S(\rho)$ for any $\rho>0$), giving
\begin{equation}
	\left[ P_{[H(E)]^+(\delta)} - P_{[H(E)]^-(\delta)} \right] (x,x;t)
	\le
	\C
	\begin{cases}
		\frac{(\log t - 2 \log R)^{1/2}R}{t^{3/2}},
		&
		R  < t^{3/7} 
		\\
		\frac{\sqrt{\delta}}{R^{5/2}},
		&
		t^{3/7} \le R
	\end{cases}
	\label{eq:HP_layer_3}
\end{equation}
(A better bound can of course be obtained from the exact expression for the Dirichlet heat kernel on the half-plane, but it would not change the final result).
This is not enough by itself to bound the summand in \cref{eq:A1_diff}, since \cref{eq:HP_layer_2} also contains a term of order $1/n^2$ which does not depend on the distance from the boundary; however,
it does give
\begin{equation}
	\left|
		\int_{E_{R_1}}
		\partial P_{H(E)} (x;n)
		\dd x
		-
		\frac12
		\sum_{x \in E_{R_1} \cap \bZ^2}
		\partial \wt P_{H(E)} (x;2n)
	\right|
	\le
	\C
	\left[
		\frac{\log n}{n^{9/14}}
		+
		\frac{R_1}{n^2}
	\right]
	,
\end{equation}
where $E_{R_1}$ is a rectangular region like $E_\kappa$ (cf.\ \cref{fig:edge}), with $R_1$ any number such that $|E_{R_1}| = \#\left( E_{R_1} \cap \bZ^2 \right)$.  
Further away from the edge, we can use \cref{thm:PD_domain_change,thm:PDt_domain_change} to bound $\partial P_{H(E)}$ and $\partial \wt P_{H(E)}$ separately, and
setting $R_1$ as close as possible to $ n^{19/14}$,
the contribution from $E \setminus E_{R_1}$ is superexponentially small for large $n$. 
In conclusion, we have
\begin{equation}
	\begin{split}
		\left|
		\frac12
		A_1(L\Pi,2n) - L \frac{a_1(\Pi)}{n^{1/2}}
		\right|
		\le
		\C 
		L
		\frac{\log n}{n^{9/14}}
	\end{split}
	\label{eq:A1_a1}
\end{equation}
for $n > 1$; note in particular that this decays faster than $1/\sqrt{n}$ for $n$ large.

For the corner contributions, we can apply the same techniques to bound $P^\angle_C(x;n) - \tfrac12 \wt P^\angle_C(y;2n)$, giving estimated in terms of distances to the vertex $v(C)$, by the same reasoning as in the derivation of \cref{eq:a2_C_rem}; so with $R_1 = n^{25/28}$ we have
\begin{equation}
	\label{eq:A2_diff}
	\begin{split}
		\left|a_2 (\Pi) - \tfrac12 A_2 (L\Pi,2n)\right|
		\le
		\C \left[ 
			\frac{\log n}{n^{3/14}}
			+
			\frac{R_1^2}{n^2}
		\right]
		\le
		\C
		\frac{\log n}{n^{3/14}}
	\end{split}
\end{equation}
for $n$ large enough.

I now define
\begin{align}
	\alpha_0
	& :=
	\log 4
	-
	\sum_{n=1}^\infty
	\frac{\wt P_0(2n)}{2n}
	=
	\log 4
	-
	\frac{1}{\# \left((L \Pi) \cap \bZ^2 \right)}
	\sum_{n=1}^\infty
	\frac{A_0(L\Pi,2n)}{2n}
	\label{eq:alpha0_def}
	, \\
	\alpha_1 (e) 
	& :=
	-
	\frac1{|e|}
	\sum_{n=1}^\infty
	\frac{\hat A_1 (E_e,1,n)}{2n}
	\label{eq:alpha1_def}
	, \\ 
	\alpha_2(c)
	& :=
	\sum_{n=1}^\infty
	\left[ 
		a_2(C_c) \int_n^{n+1} \frac{\dd t}{t}
		-
		\frac{\hat A_2 (C, 2n)}{2 n}
	\right]
	,
	\label{eq:alpha2_def}
\end{align}
for $e$ an edge of $L\Pi$, $E_e$ the associated element of $\cE(L\Pi)$, $c$ a corner and $C_c$ the associated wedge in $\cC(L\Pi)$.
Note that in light of the properties noted after \cref{eq:A2_def}, these definitions have the properties specified in points~2 and~3 of \cref{thm:main},
and furthermore
\begin{equation}
	\begin{split}
		&
		\# \left((L \Pi) \cap \bZ^2 \right) \alpha_0
		+
		\sum_{e \in \fE(L\Pi)} \alpha_1(e) |e|
		+
		\sum_{c \in \fC(L\Pi)} \alpha_2(c)
		\\ & \qquad
		=
		\# \left((L \Pi) \cap \bZ^2 \right) \log 4
		\\ & \qquad \qquad
		-
		\sum_{n=1}^\infty
		\frac1{2n}
		\left[ 
			A_0(L\Pi,2n) 
			+
			A_1(L\Pi,2n)
			+
			A_2(L\Pi,2n)
			-
			2n
			a_2(\Pi)
			\int_{n}^{n+1}
			\frac{\dd t}{t}
		\right]
		.
	\end{split}
	\label{eq:alpha_to_A}
\end{equation}
Note that \cref{eq:A0_diff,eq:A1_a1,eq:A2_diff} imply that the above sum converges; the last term in the summand formally gives the divergent integral $a_2(\Pi) \int_1^\infty \tfrac{\dd t}{t}$ whose divergence is cancelled by $A_2$.

\section{Concluding the proof of \cref{thm:main}}
\label{sec:conclusion}
Using \cref{eq:zeta_det_log,eq:zeta_M_sum},
\begin{equation}
	\log \det \wt \Delta_{L\Pi,L\Sigma}
	=
	\# \left( \Omega \cap \bZ^2 \right) 
	\log 4
	-
	\sum_{n=1}^\infty
	\frac1{2n} 
	\wt P_{L\Pi,L\Sigma} (x,x;2n)
	.
	\label{eq:det_HK_base}
\end{equation}
Recalling \cref{eq:discrete_reference_error}, 
\begin{equation}
	\begin{split}
		\sum_{n=1}^{K^2-1}
		\frac1{2n}
		\wt P_{L \Pi,L\Sigma} (x;2 n)
		=
		&
		\sum_{n=1}^{K^2-1}
		\frac1{2n}
		\left[ 
			A_0(L\Pi,2n)
			+
			A_1(L\Pi,2n)
			+
			A_2(L\Pi,2n)
		\right]
		\\ &
		+
		\cO\left( \exp\left( -\C \frac{L^2}{K^2} \right) \right)
	\end{split}
\end{equation}
for $L/K \to \infty$,
while in \cref{thm:discretization_err_integrated} we saw that
\begin{equation}
		\sum_{n=  K^2}^\infty
		\sum_{y \in L \Pi \cap \bZ^2} \frac{\wt P_{L \Pi,L\Sigma} (y,y;2n)}{2n}
		=
		\int_{K^2}^\infty
		\int_{L\Pi} P_{L \Pi,L\Sigma} (x,x;t) \dd x
		\frac{\dd t}{t}
		+
		\cO
		\left( \frac{L \log L}{K^{9/7}}\right)
\end{equation}
for $K, \ L/K, K^{9/7}/L \to \infty$.
All of the error terms will ultimately have similar forms, so I will  set $K = \ceil( \Cl{KL} L/\sqrt{\log L})$ with $\Clast$ chosen large enough so that the exponential error terms are negligible;  we then have
\begin{equation}
	\begin{split}
		\log \det \wt \Delta_{L\Pi,L\Sigma}
		= &
		\# \left( \Omega \cap \bZ^2 \right) 
		\log 4
		\\ & 
		-
		\sum_{n=1}^{K^2-1}
		\frac1{2n}
		\left[ 
			A_0(L\Pi,2n)
			+
			A_1(L\Pi,2n)
			+
			A_2(L\Pi,2n)
		\right]
		\\ &
		-
		\int_{K^2}^\infty
		\int_{L\Pi} P_{L \Pi} (x,x;t) \dd x
		\frac{\dd t}{t}
		+
		\cO\left( \frac{[\log L]^{23/14}}{L^{2/7}} \right)
	\end{split}
	\label{eq:logdet_split_1}
\end{equation}
for $L \to \infty$.

Recalling the definition of $F_\Pi$ in \cref{eq:F_Pi_def} and the relationship between $P_{\Pi,\Sigma}$ and $\Tr e^{- t \Delta_{\Pi,\Sigma}}$ (\cref{thm:DC_trace}), 
\begin{equation}
	\int_{K^2}^\infty
	\int_{L\Pi} P_{L \Pi} (x,x;t) \dd x
	\frac{\dd t}{t}
	=
	\int_{K^2}^\infty
	\left[ 
		F_{L\Pi,L\Sigma}(t/2)
		+
		\frac{a_0(L\Pi)}{t/2}
		+
		\frac{a_1(L\Pi)}{\sqrt{t/2}}
	\right]
	\frac{\dd t}{t}
\end{equation}
which, combined with \cref{eq:A0_diff,eq:A1_a1}  yields
\begin{equation}
	\begin{split}
		\int_{K^2}^\infty
		&
		\int_{L\Pi} P_{L \Pi} (x,x;t) \dd x
		\frac{\dd t}{t}
		\\ & 
		=
		\int_{K^2/2}^\infty
		F_{L\Pi,L\Sigma}(t)
		\frac{\dd t}{t}
		+
		\sum_{n=K^2+1}^\infty
		\frac{A_0(L\Pi,2n)+A_1(L\Pi,2n)}{2n}
		+
		\cO\left( \frac{[\log L]^{23/14}}{L^{2/7}} \right)
	\end{split}
	\label{eq:big_t_integral}
\end{equation}
with $K$ chosen as discussed above.
Rescaling $F$ as in \cref{eq:F_rescale} and taking advantage of the small $t$ asymptotics in \cref{eq:ht_asympt}, we have
\begin{equation}
	\begin{split}
		\int_{K^2/2}^\infty
		F_{L\Pi,L\Sigma}(t)
		\frac{\dd t}{t}
		= &
		\int_{0}^\infty
		F_{\Pi,\Sigma}(t)
		\frac{\dd t}{t}
		+
		\int_{K^2/2}^{e^{-\gamma}L^2}
		a_2(L\Pi)
		\frac{\dd t}{t}
		+
		\cO\left( \exp\left( -\C \frac{L^2}{K^2} \right) \right)
		\\ = &
		\int_{0}^\infty
		F_{L\Pi,L\Sigma}(t)
		\frac{\dd t}{t}
		-
		\int_1^{K^2}
		a_2(L\Pi)
		\frac{\dd t}{t}
		+
		\cO\left( \exp\left( -\Clast \frac{L^2}{K^2} \right) \right)
	\end{split}
\end{equation}
and combining this with \cref{eq:big_t_integral,eq:logdet_split_1} gives
\begin{equation}
	\begin{split}
		\log \det \wt \Delta_{L\Pi,L\Sigma}
		= &
		\# \left( \Omega \cap \bZ^2 \right) 
		\log 4
		-
		\sum_{n=1}^{\infty}
		\frac1{2n}
		\left[ 
			A_0(L\Pi,2n)
			+
			A_1(L\Pi,2n)
		\right]
		\\ &
		-
		\sum_{n=1}^{K^2}
		\frac{1}{2n}
		\left[ 
			A_2(L\Pi,2n)
			-
			2 n a_2(\Pi)
			\int_n^{n+1} \frac{\dd t}{t}
		\right]
		\\ & +
		\int_{0}^\infty
		F_{L \Pi,L\Sigma} (t)
		\frac{\dd t}{t}
		+
		\cO\left( \frac{[\log L]^{23/14}}{L^{2/7}} \right)
		.
	\end{split}
\end{equation}
Finally, using \cref{eq:A2_diff} to bound the error made in extending the last sum to infinity, and recognizing $\alpha_0, \ \alpha_1$ and $\alpha_2$ using \cref{eq:alpha_to_A} and $\alpha_4(L\Pi,L\Sigma) = \zeta'_{\Delta_{L\Pi,L\Sigma}}(0)$ using \cref{eq:z_reg_Delta_final}, we obtain \cref{eq:main_alt}.  As mentioned before, \cref{thm:main} then follows by using \cref{eq:zeta_rescaling}.

\appendix
\section{The Laplacian with Dirichlet boundary conditions on a double cover}
\label{app:DC}

Let $\Pi_\Sigma$ denote a double cover of $\Pi$ branched around each point in $\Sigma$, and let $p: \Pi_\Sigma \to \Pi$ be the associated projection map.  
Let $L^2(\Pi,\Sigma)$ be the set of equivalence classes (with respect to almost-everywhere equality) of square-integrable functions $f : \Pi_\Sigma \to \bR$ with the property that $f(x) = - f(y)$ if $p(x) = p(y)$ but $x \neq y$.
$L^2(\Pi,\Sigma)$ is evidently a closed linear subspace of $L^2(\Pi_\Sigma)$ and so is itself a Hilbert space with the same inner product.
Let $C_0^\infty(\Pi,\Sigma)$ be the set of smooth functions $f \in L^2(\Pi,\Sigma)$ such that, for all $y \in \partial \Pi_\Sigma$, $\lim_{x \to y} f(x) = 0$ and the corresponding limits of all of the derivatives of $f$ exist.

The Laplacian $\Delta_{\Pi,\Sigma}$ is defined as the Friedrichs extension of the (positive) Laplacian $- \nabla \cdot \nabla $ on $C_0^\infty$, which is self-adjoint with respect to the inner product of $L^2(\Pi,\Sigma)$.  It is a strictly positive operator: letting $\Delta_\Pi$ be the Dirichlet Laplacian on $\Pi$ (not the double cover), for any $f \in C_0^\infty (\Pi,\Sigma) $ the function $\tilde f$ defined by $\tilde f(p(x)) = |f(x)|$ is in the domain of $\Delta_\Pi$; we can extend it to a function in the domain of $\Delta_R$ for $R$ a rectangle containing $\Pi$, and then
\begin{equation}
	\frac{(f,\Delta_{\Pi,\Sigma}f)}{(f,f)}
	\ge
	\frac{(\tilde f, \Delta_\Pi \tilde f)}{(\tilde f, \tilde f)}
	=
	\frac{(\tilde f, \Delta_R \tilde f)}{(\tilde f, \tilde f)}
	\label{eq:D_domain_monotone}
\end{equation}
where in each expression $(\cdot,\cdot)$ denotes the appropriate $L_2$ inner product;
of course the spectrum of $\Delta_R$ is explicitly known and strictly positive, so the last ratio is uniformly bounded from below by a positive number, a bound which thus extends to all $f$ in the domain of $\Delta_{\Pi,\Sigma}$.
Among other things, this combines with the spectral theorem to imply that $e^{-t \Delta_{\Omega,\Sigma}}$  is defined for all $t \ge 0$.

For $\cW_\cdot$ a Brownian trajectory on $\bR^2$ which does not intersect $\Sigma$ (which is finite, so these trajectories form a full measure set), let $\cW_\cdot^\Sigma$ denote its continuous lift to the dual cover of $\bR^2$ branched around $\Sigma$, starting from a specified starting point $x \in \Pi_\Sigma$.  Consider the family of operators $K_t$ on $L^2(\Pi,\Sigma)$ defined by
\begin{equation}
	[K_t f](x)
	=
	\bE_x \left[ f\left( \cW_t^\Sigma \right) \ind_{[t,\infty)}(T_\Pi) \right]
	.
	\label{eq:DC_kernel_start}
\end{equation}
For $f \in C_0^\infty(\Pi,\Sigma)$, by definition $f(\cW^\Sigma_{T_\Pi})=0$, and so 
$f\left( \cW^\Sigma_t \right)\ind_{[t,\infty)}(T_\Pi) = f(\cW^\Sigma_{t \wedge T_\Pi})$;
bearing in mind that $\Delta_{\Pi,\Sigma}$ acts locally on smooth functions like the ordinary Laplacian, It\^o's formula implies that
\begin{equation}
	\frac{\partial}{\partial t} [K_t f](x)
	=
	- \bE_x \left[\tfrac12 \Delta_{\Pi,\Sigma} f\left( \cW_{t \wedge T_\Pi}^\Sigma \right)  \right]
	=
	- [\tfrac12 K_t \Delta_{\Pi,\Sigma} f](x)
	;
	\label{eq:Kt_Ito}
\end{equation}
by continuity, this implies $\tfrac{\dd}{\dd t} K_t = - \tfrac12 K_t \Delta_{\Pi,\Sigma}$ as operators on $L^2(\Pi,\Sigma)$, and evidently $K_0$ is the identity; in other words $K_t$ is a one-parameter group generated by $- \tfrac12 \Delta_{\Pi,\Sigma}$; there is only one such group \cite[Theorem~1.7]{Davies.1param}, so $K_t = e^{- (t/2) \Delta_{\Pi,\Sigma}}$.

\begin{theorem}
	For any closed, simply connected $B \subset \Pi$, there is a function $Q_B(x,y;t)$ from $\Pi_\Sigma \times p^{-1}(B) \times (0,\infty)$ to $\bR$, such that
	\begin{equation}
		\left[ e^{-(t/2) \Delta_{\Pi,\Sigma}} g \right] (x)
		=
		\int_{p^{-1}(B)}
		Q_B(x,y;t)
		g(y)
		\dd y
		\label{eq:Q_kernel_1}
	\end{equation}
	for any $g \in L^2(\Pi,\Sigma)$ with $\supp g \subset B$;
	furthermore $Q_B$ is $C^\infty$ as a function of $y$, and
	\begin{equation}
		0 \le Q_{B} (x,y;t) \le P (p(x),p(y);t)
		,
		\label{eq:QB_bounded}
	\end{equation}
	where $P$ is the full-plane heat kernel introduced in \cref{eq:P_base}.
	\label{lem:local_DC_kernel}
\end{theorem}
\begin{proof}
	Let $B_1,B_2 \subset \Pi$ be open and simply connected with ${\lis B_1 \subset B \subset B_2}$ and $\delta := d(B,\partial B_1) \vee d(B_1, \partial B_2) > 0$.
	Letting $\cW$ denote a planar Brownian motion started at $x$, I introduce two related sequences of hitting times as follows: let $\tau_0 = \sigma_0 = 0$, and for $j=1,2,\dots$
	\begin{gather}
		\sigma_j 
		 :=
		\inf \left\{ \sigma \ge \tau_{j-1} : \ \cW_\sigma \in B_1, \ \sigma < T_\Pi \right\}
		\\
		\tau_j 
		 := 
		\inf \left\{  \tau > \sigma_j : \ \cW_\tau \notin B_2 \right\}
		;
	\end{gather}
	with these definitions (thanks in particular to the nonzero distance between $B_1$ and the boundary of $B_2$), $J := \max \{j = 0,1,\dots : \ \sigma_j \le t\}$ is almost surely finite.
	Furthermore
	\begin{equation}
		g\left( \cW_t^\Sigma \right)
		\ind\left\{ T_\Pi > t  \right\}
		=
		g\left( \cW_t^\Sigma \right)
		\ind\left\{ J \ge 1  \right\}
		\ind\left\{ \tau_J > t \right\}
	\end{equation}
	and so, recalling \cref{eq:DC_kernel_start}, 
	\begin{equation}
		\begin{split}
			\left[ e^{-(t/2) \Delta_{\Pi,\Sigma}} g \right] (x)
			&=
			\bE_x\left[ 
				g\left( \cW_t^\Sigma \right)
				\ind\left\{ J \ge 1  \right\}
				\ind\left\{ \tau_J > t \right\}
			\right]
			\\ & =
			\bE_x\left[ 
				\ind\left\{ J \ge 1  \right\}
				G\left( \cW_{\sigma_J}^\Sigma; t- \sigma_J \right)
			\right]
		\end{split}
		\label{eq:DC_kernel_1}
	\end{equation}
	using the strong Markov property, where 
	\begin{equation}
		G(y;s) 
		:=
		\bE_{y}\left[ g\left( \cW_s^{\Sigma} \right) \ind\left\{ T_{B_2 }> s \right\} \right]
		.
		\label{eq:DC_Gys_def}
	\end{equation}
	Since $B_2$ is simply connected, thanks to the indicator function $G(y;s)$ is an integral over paths which remain in the same connected component of $p^{-1}(B_2)$ as $y=\cW_{\sigma_J}^\Sigma$.  To spell this out, I introduce the following notation:  Let $B_2^1$ and $B_2^2$ be the two connected components of $B_2$, 
	\begin{equation}
		\kappa(x) := \begin{cases}
			j , & x \in B_2^j
			\\
			0, & p(x) \notin B_2
		\end{cases}
		\label{eq:DC_kappa_def}
	\end{equation}
	and $g_{\kappa}:B_2 \to \bR$, $x \mapsto g \left( [p |_{B_2^\kappa}]^{-1}(x) \right)$.
	Then
	\begin{equation}
		\begin{split}
			G(y;s) 
			&= 
			\bE_{p(y)} \left[ g_{\kappa(y)}\left( \cW_s \right) \ind\left\{ T_{B_2}>s \right\}\right]
			\\ & =
			\int_{p^{-1}(B)} 
			\ind \left\{ z \in B_2^{\kappa (y)} \right\}
			g(z) P_{B_2} (p(y),p(z);s) \dd z
		\end{split}
	\end{equation}
	where $P_{B_2}$ is the Dirichlet heat kernel as introduced in \cref{eq:PD_Om_def}.
	Returning to \cref{eq:DC_kernel_1}, we then have
	\begin{equation}
		\left[ e^{-(t/2) \Delta_{\Pi,\Sigma}} g \right] (x)
		=
		\int_{p^{-1}(B)} 
		g(z)
		Q_B(x,z;t)
		\dd z
		\label{eq:QB_kernel_intro}
	\end{equation}
	with
	\begin{equation}
		\begin{split}
			&
			Q_B(x,z;t)
			 :=
			\bE_x\left[ 
				\ind\left\{ J \ge 1 \right\}
				\ind \left\{ z \in B_2^{\kappa (\cW_{\sigma_J}^\Sigma)} \right\}
				P_{B_2} (\cW_{\sigma_J},p(z); t- \sigma_J)
			\right]
			\\ & \ = 
			\bE_x\Big[ 
				\ind\left\{ J \ge 1 \right\}
				\ind \left\{ z \in B_2^{\kappa (\cW_{\sigma_J}^\Sigma)} \right\}
			\begin{aligned}[t]
				\Big(
				& P\left( \cW_{\sigma_J}, p(z); t - \sigma_J \right)
				\\& -  \bE_{\cW_{\sigma_J}}\left[ 
					P\left( \cW_{T_{B_2}}, p(z); t- \sigma_J - t_{B_2} \right)
				\right]
			\Big) \Big]
			.
			\end{aligned}
		\end{split}
		\label{eq:QB_def}
	\end{equation}
	The bounds in \cref{eq:QB_bounded} are evident from the first form of this definition, since $P_{B_2}$ and $P$ are non-negative.
	Note also that the last expression in \cref{eq:QB_def} is an integral over values of $P(y,p(z);s)$ with either $s=t$ (if $\sigma_J =0$) or $y \in \partial B_1 \cup \partial B_2 $ (and thus $|y - p(z)| \ge \delta$).  With the other arguments so restricted, $P$ is uniformly smooth as a function of $z$, and so we also see that $Q_B$ is smooth as claimed.
\end{proof}

Given that they are smooth in the second argument, any two kernels $Q_B(x,\cdot;y)$ and $Q_{B'}(x,\cdot;t)$ defined on different domains coincide on the intersection of their domains, and so we can extend by linearity to the rest of $L^2(\Pi,\Sigma)$ and obtain
\begin{cor}
	There is a function $Q_{\Pi,\Sigma} : \Pi_\Sigma \times \Pi_\Sigma \times (0,\infty) \to \bR$ such that
	\begin{equation}
		\left[ e^{-(t/2) \Delta_{\Pi,\Sigma}} g \right] (x)
		=
		\int_{\Pi_\Sigma}
		Q_{\Pi,\Sigma}(x,y;t)
		f(y)
		\dd y
		\label{eq:QPS_kernel}
	\end{equation}
	for any $f \in L^2(\Pi,\Sigma)$;
	furthermore $Q_{\Pi,\Sigma}$ is $C^\infty$ as a function of $y$, and
	\begin{equation}
		0 \le Q_{\Pi,\Sigma} (x,y;t) \le P (p(x),p(y);t)
		.
		\label{eq:QPS_bound}
	\end{equation}
	\label{cor:QPS}
\end{cor}
With this result in hand, I am ready to present the
\begin{proof}[Proof of \cref{thm:DC_trace}]
	For any simply-connected, Borel-measurable $U \subset \Pi$, let ${f_U:\Pi_\Sigma \to \bR}$ be either of the two functions which is equal to 1 on one of the connected components of $p^{-1}(U)$, -1 on the other, and zero elsewhere.  Evidently ${f_U \in L^2(\Pi,\Sigma)}$, $(f_U,f_U) = 2|U|$, and
	\begin{equation}
		\begin{split}
			\left( f_U, e^{-(t/2) \Delta_{\Pi,\Sigma}} f_U \right)
			& =
			\int_{p^{-1}(U)\times p^{-1}(U)} Q_{\Pi,\Sigma}(x,y;t) \gamma_U(x,y) \dd x \dd y
			\\ & =
			\int_{p^{-1}(U)} \bE_x\left[ \gamma_U (\cW_0^\Sigma,\cW^\Sigma_t) \ind\left\{ \cW_t \in U \right\} \ind\left\{ T_\Pi > t \right\} \right]
			\dd x
		\end{split}
		\label{eq:fU_exp_IP}
	\end{equation}
	where
	\begin{equation}
		\gamma_U(x,y)
		=
		f_U(x) f_U(y)
		=
		\begin{cases}
			1
			, &
			\begin{aligned}
			& x \text{ and } y \text{ are in the same} \\ &\text{connected component of } p^{-1}(U)
			\end{aligned}
			\\
			-1
			, &
			\text{otherwise,}
		\end{cases}
		\label{eq:DC_gamma}
	\end{equation}
	which is independent of which $f_U$ was chosen.  Note also that if $\diam U$ is small enough, in the last integral $ \gamma(\cW^\Sigma_0,\cW^\Sigma_t)$ is equivalent to $e^{i \pi W_\Sigma(\cW;t)}$ (recall that $W_\Sigma$ was defined in the statement of \cref{thm:DC_trace} to take integer values).
	Examining \cref{eq:PD_Sig} in light of this, we see that for $x \in \Pi_\Sigma$ the definition of $P_{\Pi,\Sigma}$ there gives
	\begin{equation}
		\begin{split}
			P_{\Pi,\Sigma}(p(x),p(x);t)
			& =
			\lim_{r \to 0^+}
			\frac{1}{\pi r^2}
			\int_{p^{-1}(B_r(x))}
			\gamma_{B_r(x)}(x,y)
			Q_{\Pi,\Sigma}(x,y;t)
			\\ & =
			\sum_{y : \ p(y) = p(x)}
			\begin{cases}
				Q_{\Pi,\Sigma}(x,y,;t)
				, & 
				x = y
				\\
				- Q_{\Pi,\Sigma}(x,y,;t)
				, & 
				x \neq y
			\end{cases}
		\end{split}
		\label{eq:PD_Sig_to_Q}
	\end{equation}
	thanks to the smoothness of $Q_{\Pi,\Sigma}$ shown in \cref{cor:QPS}.

	Let $\cU_N$, for $N$ a positive integer, denote the nonempty sets which coincide with a connected component of a set of the form $\Pi \cap 2^{-N} \left( [m,m+1) \times [n,n+1) \right)$ for some $m,n \in \bZ$; then $\cU_N$ is a finite set of disjoint, simply connected Borel sets whose union is $\Pi$.
	Numbering the elements of each $\cU_N$ arbitrarily, so that
	\begin{equation}
		\cU_N = \left\{ U_{N1}, \dots, U_{N|\cU_N|}\right\},
	\end{equation}
	the sequence $f_{U_{11}},\dots,f_{U_{1|\cU_1|}},f_{U_{21}},\dots$ is an overcomplete basis of $L^2(\Pi,\Sigma)$ and by applying the Hilbert-Schmidt process to this sequence we obtain an orthonormal basis $\psi_1,\psi_2,\dots$ such that
	\begin{equation}
		\sum_{j=1}^{|\cU_N|} 
		(\psi_j, e^{- t \Delta_{\Pi,\Sigma}} \psi_j)
		=
		\sum_{U \in \cU_N} \frac{(f_U, e^{-t \Delta_{\Pi,\Sigma}}f_U)}{(f_U,f_U)}
	\end{equation}
	for all positive integer $N$.  Since $e^{-t \Delta_{\Pi,\Sigma}}$ is non-negative, this suffices to show that
	\begin{equation}
		\Tr e^{-t \Delta_{\Pi,\Sigma}}
		=
		\sum_{j=1}^\infty 
		(\psi_j, e^{- t \Delta_{\Pi,\Sigma}} \psi_j)
		=
		\lim_{N \to \infty }
		\sum_{U \in \cU_N} \frac{(f_U, e^{-t \Delta_{\Pi,\Sigma}}f_U)}{(f_U,f_U)}
		.
		\label{eq:DC_trace_to_lim_sum}
	\end{equation}
	Using \cref{eq:fU_exp_IP}, 
	\begin{equation}
		\sum_{U \in \cU_N} \frac{(f_U, e^{-t \Delta_{\Pi,\Sigma}}f_U)}{(f_U,f_U)}
		=
		\frac12
		\int_{\Pi_\Sigma}\left[ 
			\frac{1}{|U_N(x)|}
			\int_{p^{-1}(U_N(x))}
			\gamma_{U_N(x)}(x,y)
			Q_{\Pi,\Sigma}(x,y;2t)
			\dd y
		\right]
		\dd x
		,
	\end{equation}
	where $U_N(x)$ is the unique element of $\cU_N$ which includes $x$.
	As $N \to \infty$, the quantity in brackets converges pointwise to the same limit as in \cref{eq:PD_Sig_to_Q} with $t$ replaced by $2t$, i.e.\ to $P_{\Pi,\Sigma}(p(x),p(x);2t)$, and this together with the bound in \cref{eq:QPS_bound}
	 provide the conditions to apply the dominated convergence theorem and obtain
	 \begin{equation}
		\Tr e^{-t \Delta_{\Pi,\Sigma}}
		=
		\frac12
		\int_{\Pi_\Sigma}
		P_{\Pi,\Sigma}(p(x),p(x);2 t)
		\dd x
	 \end{equation}
	 which, passing to an integral over $x \in \Pi$, is equivalent to \cref{eq:trace_PD}.
\end{proof}

\section{A simple bound on the distribution of a simple random walk}
\begin{lemma}
	For all integers $n \ge 1$ and $0 \le m \le n$,
	\begin{equation}
		2^{-n}
		\binom{n}{m}
		\le
		\frac{\C}{\sqrt n}
		\exp\left( -\frac{(2m-n)^2}{2n} \right)
		.
		\label{eq:RW_Gaussian}
	\end{equation}
	\label{lem:RW_Gaussian}
\end{lemma}

\begin{proof}
	Firstly, note that if $m=0$ or $m=n$, the bound is obvious, so we consider $1 \le n \le n-1$.
	In this case, using Stirling's formula with explicit error estimates \cite{Robbins},
	\begin{equation}
		\begin{split}
			\log \binom{n}{m}
			& \le
			- \tfrac12 \log 2\pi + (n+1/2) \log n - (m + 1/2) \log m 
			\\ & \quad \quad
			- (n-m+1/2) \log (n-m)
			+ \frac{1}{12n}
			.
		\end{split}
		\label{eq:Stirling}
	\end{equation}
	Letting $L_n(\mu) = (\mu + 1/2) \log \mu + (n- \mu+1/2)\log (n-\mu)$, we have
	$
		L'_n(n/2) = 0
		$,
	\begin{equation}
		L''_n (\mu)
		=
		\frac{1}{\mu}
		+
		\frac{1}{n-\mu}
		-
		\frac{1}{2\mu^2}
		-
		\frac{1}{2(n-\mu)^2}
		\label{eq:Lpp}
	\end{equation}
	and
	\begin{equation}
		\begin{split}
			L'''_n(\mu)
			& =
			-\frac{1}{\mu^2}
			+\frac{1}{(n-\mu)^2}
			+\frac{1}{\mu^3}
			-\frac{1}{(n-\mu)^3}
			\\ & =
			\frac{2\mu-n}{\mu^3(n-\mu)^3}
			\left[ (n+1)\mu^2 - (n^2+n)\mu + n^2 \right]
			;
		\end{split}
		\label{eq:Lp3}
	\end{equation}
	evidently $n/2$ is the unique local minimum of $L''_n(\mu)$. It is easy to see that
	\begin{equation}
		L''_n(1)
		=L''_n(n-1)
		=
		\frac12 + \frac{1}{n-1} - \frac{1}{2(n-1)^2}
		>
		L''_n(n/2)
		=
		\frac{4}{n}
		-
		\frac{4}{n^2}
	\end{equation}
	for $n$ large enough (in fact for $n \ge 4$), and so we have
	\begin{equation}
		\begin{split}
			L_n (\mu) 
			& \ge 
			L_n (n/2) 
			+
			L''_n (n/2) \frac{(m-n/2)^2}{2}
			\\ & =
			(n+1) \log (n/2)
			+
			\frac{(2m-n)^2}{2 n}
			-
			\frac{(2m-n)^2}{2 n^2}
			.
		\end{split}
		\label{eq:Ln_bound}
	\end{equation}
	Noting that the last term is bounded by $1/2$, this combines with \cref{eq:Stirling} to give
	\begin{equation}
		\log \binom{n}{m}
		\le
		n \log 2
		- \frac{(2m-n)^2}{2 n}
		- \frac12 \log n
		+
		\C
	\end{equation}
	for sufficiently large $n$, and the result can be trivially extended to smaller $n$.
\end{proof}

\section*{Acknowledgements}

I thank Bernard Duplantier for introducing me to this problem, and pointing out its relationship to the heat kernel.  The outline of the proof occurred to me during the workshop ``Dimers, Ising Model, and their Interactions'' at the Banff International Research Station, and I thank the organizers for the invitation to attend the workshop as well as the many participants in the discussions which provoked it.  I would also like to thank Alessandro Giuliani for discussing the work with me through all the stages of development, and the anonymous referee whose comments led to a considerable number of corrections and improvements.

This work, a major part of which was conducted at the Department of Mathematics and Physics of the Università degli Studi Roma Tre (Rome, Italy), was supported by the European Research Council (ERC) under the European Union's Horizon 2020 research and innovation programme (ERC CoG UniCoSM, grant agreement n.724939 and ERC StG MaMBoQ, grant agreement No. 802901).  The final stages were also partially supported by  the MIUR Excellence Department Project MatMod@TOV awarded to the Department of Mathematics, University of Rome Tor Vergata.

\section*{Data availability}

Data sharing is not applicable to this article as no new data were created or analyzed in this study.
\printbibliography[heading=bibintoc]
\end{document}